\documentclass[sigconf]{acmart}
\pdfoutput=1
\usepackage{booktabs}
\usepackage[noend]{algpseudocode}
\usepackage{algorithm,algorithmicx}
\usepackage{mathtools}
\usepackage{multirow}
\usepackage{caption}
\usepackage{balance}
\usepackage[flushleft]{threeparttable}
\usepackage{eqnarray}
\usepackage{graphicx}

\newcommand{\tribrackets}[1]{\left\langle #1 \right\rangle}
\newcommand{\encdec}[3]{\tribrackets{#1,\, #2, \,#3}}
\newcommand{\bmaps}{\phi,\psi,\upsilon}
\newcommand{\bmapss}{\phi, \, \psi, \, \upsilon}
\newcommand{\strass}[4]{\tribrackets{#1, #2, #3 ; #4}}
\newcommand{\strassalt}[7]{\strass{#1}{#2}{#3}{#4} _{#5,#6,#7}}
\newcommand{\strassaltdef}[4]{\strass{#1}{#2}{#3}{#4} _{\bmaps}}

\newcommand{\squarestrass}[2]{\strass{#1}{#1}{#1}{#2}}

\newcommand{\nnz}[1]{nnz\left(#1 \right)}
\newcommand{\nns}[1]{nns\left(#1 \right)}

\def\OmValid {$\Omega$-valid}
\def\OmValidity {$\Omega$-validity}
\def\OmValidator {$\Omega$-validator}
\def\OmIndep {$\Omega$-independent}

\theoremstyle{acmplain}
\newtheorem{claim}[theorem]{Claim}
\theoremstyle{acmdefinition}
\newtheorem{problem}[theorem]{Problem}
\newtheorem{remark}[theorem]{Remark}
\newtheorem{notation}[theorem]{Notation}
\newtheorem{fact}[theorem]{Fact}

\theoremstyle{acmplain}

\begin{document}
\acmYear{2020}
\acmConference{}{}{}
\acmPrice{}
\acmDOI{}
\acmISBN{}
\setcopyright{none}
\fancyhead{}

\title{Sparsifying the Operators of Fast Matrix Multiplication Algorithms}


\author{Gal Beniamini}
\affiliation{
	\institution{The Hebrew University of Jerusalem}
}
\email{gal.beniamini@mail.huji.ac.il}

\author{Nathan Cheng}
\affiliation{
	\institution{University of California at Berkeley}
}
\email{ncheng@berkeley.edu}

\author{Olga Holtz}
\affiliation{
	\institution{University of California at Berkeley}
}
\email{holtz@math.berkeley.edu}

\author{Elaye Karstadt}
\affiliation{
  \institution{The Hebrew University of Jerusalem}
}
\email{elaye.karstadt@mail.huji.ac.il}

\author{Oded Schwartz}
\affiliation{
  \institution{The Hebrew University of Jerusalem}
}
\email{odedsc@cs.huji.ac.il}

\settopmatter{printfolios=true}


\begin{abstract}
Fast matrix multiplication algorithms may be useful, provided that their running time is good in practice. Particularly, the leading coefficient of their arithmetic complexity needs to be small. Many sub-cubic algorithms have large leading coefficients, rendering them impractical. Karstadt and Schwartz (SPAA'17, JACM'20) demonstrated how to reduce these coefficients by sparsifying an algorithm's bilinear operator. Unfortunately, the problem of finding optimal sparsifications is NP-Hard.

We obtain three new methods to this end, and apply them to existing fast matrix multiplication algorithms, thus improving their leading coefficients. These methods have an exponential worst case running time, but run fast in practice and improve the performance of many fast matrix multiplication algorithms. Two of the methods are guaranteed to produce leading coefficients that, under some assumptions, are optimal.

\end{abstract}

\maketitle

\section{Introduction}

Matrix multiplication is a fundamental computation kernel, used in many fields ranging from imaging to signal processing and artificial neural networks. The need to improve performance has attracted much attention from the science and engineering communities. Strassen's discovery of the first sub-cubic algorithm~\cite{strassen1969gaussian} sparked intensive research into the complexity of matrix multiplication algorithms (cf.~\cite{winograd1971multiplication, hopcroft1971minimizing, laderman1976noncommutative, pan1978strassen, de1978varieties1, de1978varieties2, bini1979n,schonhage1981partial,romani1982some,pan1982trilinear, coppersmith1982asymptotic, strassen1986asymptotic, johnson1986noncommutative,coppersmith1990matrix, laderman1992practical, kaporin1999practical, cohn2003group,grolmusz2008modular,stothers2010complexity,williams2012multiplying, smirnov2013bilinear,le2014powers, benson2015framework, smirnov2017several, karstadt2017matrix, beniamini2019faster, karstadt2020matrix}).

The research efforts can be divided into two main branches. The first revolves around the search for asymptotic upper bounds on the arithmetic complexity of matrix multiplication (cf.~\cite{coppersmith1982asymptotic,
strassen1986asymptotic,coppersmith1990matrix, cohn2003group,grolmusz2008modular,stothers2010complexity,williams2012multiplying,le2014powers}). This approach focuses on asymptotics, typically disregarding the hidden constants of the algorithms and other aspects of practical importance. Many of these algorithms remain highly theoretical due to their large hidden constants, and furthermore, they apply only to matrices of very high dimensions.

In constrast, the second branch focuses on obtaining matrix multiplication algorithms that are both asymptotically fast and practical. This requires the algorithms to have reasonable hidden constants that are applicable even to small instances (cf.,~\cite{winograd1971multiplication, hopcroft1971minimizing, laderman1976noncommutative, pan1978strassen, bini1979n,schonhage1981partial,romani1982some,pan1982trilinear, johnson1986noncommutative, laderman1992practical, kaporin1999practical, smirnov2013bilinear, benson2015framework, smirnov2017several, karstadt2017matrix, beniamini2019faster, karstadt2020matrix}).

\begin{table*}[t]
	\centering
	\caption {Examples of improved leading coefficients} \label{tab:ks-compare}
	\small{
	\begin{tabular}{ |l||c||c|c|c||c|c|c||c|c|| }
		\hline
		\multirow{3}{*}{\ \ Algorithm} &
		\multirow{3}{*}{Leading Monomial} &
		\multicolumn{3}{|c||}{Arithmetic Operations} &
		\multicolumn{3}{c||}{Leading Coefficient} &
		\multicolumn{2}{c||}{Improvement} \\
		\cline{3-10}
		& & Original & \cite{karstadt2017matrix, karstadt2020matrix} & Here & Original & \cite{karstadt2017matrix, karstadt2020matrix} & Here & \cite{karstadt2017matrix, karstadt2020matrix} & Here \\
		\hline 
		&&&&&&&&&\\[-1em]
		$\langle 2,2,2;7 \rangle$~\cite{strassen1969gaussian} &$n^{\log_{2} 7}\approx n^{2.80735}$ & 18 & 12 & 12 & 7 & 5 & 5 & 28.57\% &28.57\%\\
		\hline 
		&&&&&&&&&\\[-1em]
		$\langle 3,2,3;15 \rangle$~\cite{benson2015framework} &$n^{\log_{18} 15^{3}}\approx n^{2.81076}$ & 64 & 52 & 39 & 9.61 & 7.94 & 6.17 & 17.37\% &35.84\%\\
		\hline 
		&&&&&&&&&\\[-1em]
		$\langle 4,2,3;20 \rangle$~\cite{smirnov2013bilinear} &$n^{\log_{24} 20^{3}}\approx n^{2.82789}$ & 78 & 58 & 51 & 8.9 & 7.46 & 5.88 & 16.17\% &33.96\%\\
		\hline 
		&&&&&&&&&\\[-1em]
		$\langle 3,3,3;23 \rangle$~\cite{benson2015framework} &$n^{\log_{3} 23}\approx n^{2.85404}$ & 87 & 75 & 66 & 7.21 & 6.57 & 5.71 & 8.87\%& 20.79\%\\
		\hline 
		&&&&&&&&&\\[-1em]
		$\langle 6,3,3;40 \rangle$~\cite{smirnov2013bilinear} & $n^{\log_{54} 40^{3}}\approx n^{2.77429}$ & 1246 & 202 & 190 & 55.63 & 9.36 & 8.9 & 83.17\% &84.01\%\\
		\hline
	\end{tabular}
	\caption*{\small {The leading monomial of rectangular $\strass{n}{m}{k}{t}$-algorithms refers to their composition~\cite{hopcroft1973duality} into square $\strass{nmk}{nmk}{nmk}{t^{3}}$-algorithms. The improvement column is the ratio between the new and the original leading coefficients of the arithmetic complexity. See Table~\ref{tab:decomposedalgs} for a full list of results.}}
	}
\end{table*}

\subsection{Previous work}

\paragraph{Reducing the leading coefficients.} Winograd~\cite{winograd1971multiplication} reduced the leading coefficient of Strassen's algorithm's arithmetic complexity from 7 to 6 by decreasing the number of additions and subtractions in the $2\times 2$ base case from 18 to 15\footnote{See Section~\ref{subsec:fmm} for the connection between the number of additions and the leading coefficient.}. Later, Bodrato~\cite{bodrato2010strassen} introduced the intermediate representation method, that successfully reduces the leading coefficient to 5, for repeated squaring and chain matrix multiplication. Cenk and Hasan~\cite{cenk2017arithmetic} presented a non-uniform implementation of Strassen-Winograd's algorithm~\cite{winograd1971multiplication}, which also reduces the leading coefficient from 6 to 5, but incurs additional penalties such as a larger memory footprint and higher communication costs. Independently, Karstadt and Schwartz~\cite{karstadt2017matrix, karstadt2020matrix} used a technique similar to Bodrato's, and obtained a matrix multiplication algorithm with a $2\times 2$ base case, using 7 multiplications, and a leading coefficient of 5. Their method also applies to other base cases, improving the leading coefficients of multiple algorithms. Beniamini and Schwartz~\cite{beniamini2019faster} introduced the decomposed recursive bilinear framework, which generalizes~\cite{karstadt2017matrix, karstadt2020matrix}. Their technique allows a further reduction of the leading coefficient, yielding several fast matrix multiplication algorithms with a leading coefficient of 2, matching that of the classical algorithm.

\paragraph{Lower bounds on leading coefficients. }
Probert~\cite{probert1976additive} proved that 15 additions are necessary for any recursive-bilinear matrix multiplication algorithm with a $2\times2$ base case using 7 multiplications over $\mathbb{F}_{2}$, which corresponds to a leading coefficient of 6. This was later matched by Bshouty \cite{bshouty1995additive}, who used a different technique to obtain the same lower bound over an arbitrary ring. Both cases have been interpreted as a proof of optimality for the leading coefficient of Winograd's algorithm~\cite{winograd1971multiplication}.

Karstadt and Schwartz's $\squarestrass{2}{7}$-algorithm\footnote{See Section \ref{subsec:fmm} for definition.}~\cite{karstadt2017matrix, karstadt2020matrix} requires 12 additions (thus having a leading coefficient of 5) and seemingly contradicts these lower bounds. Indeed, they showed that these lower bounds~\cite{probert1976additive, bshouty1995additive} do not hold under alternative basis multiplication. In addition, they extended the lower bounds to apply to algorithms that utilize basis transformations, and showe that 12 additions are necessary for any recursive-bilinear matrix multiplication algorithm with a $2\times2$ base case using 7 multiplications, regardless of basis. Thus proving a lower bound of 5 on the leading coefficient of such algorithms.

Beniamini and Schwartz~\cite{beniamini2019faster} extended the lower bound to the generalized setting, in which the input and output can be transformed to a basis of larger dimension. They also found that the leading coefficient of any such algorithm with a $2 \times 2$ base case using 7 multiplications is at least 5.

\paragraph{Obtaining alternative basis algorithms.}
Recursive-bilinear algorithms can be described by a triplet of matrices, dubbed the encoding and decoding matrices (see Section~\ref{subsec:fmm}). The alternative basis technique~\cite{karstadt2017matrix, karstadt2020matrix} utilizes a decomposition of each of these matrices into a \textit{pair} of matrices -- a basis transformation, and a \textit{sparse} encoding or decoding matrix. Once a decomposition is found, applying the algorithm is straightforward (see Section~\ref{subsec:fmm}).

The leading coefficient of the arithmetic complexity is determined by the number of non-zero (and non-singleton) entries in each of the encoding/decoding matrices, while the basis transformations only affect the low order terms of the arithmetic complexity (see Section~\ref{subsec:fmm}). Thus, reducing the leading coefficient of fast matrix multiplication algorithms translates to the matrix sparsification (MS) problem.

\paragraph{Matrix sparsification. }
Unfortunately, matrix sparsification is NP-Hard to solve~\cite{mccormick1983combinatorial} and NP-Hard to approximate to within a factor of $2^{\log^{.5 - o\left(1\right)} n}$~\cite{gottlieb2010matrix} (Over $\mathbb{Q}$, assuming NP does not admit quasi-polynomial time deterministic algorithms). Despite the problem being NP-hard, search heuristics can be leveraged to obtain bases which significantly sparsify the encoding/decoding matrices of fast matrix multiplication algorithms with small base cases.

There are a few heuristics that can solve the problem, under severe assumptions, such as the full rank of any square submatrix, and requiring that the rank of each submatrix be equal to the size of the largest matching in the induced bipartite graph (cf.,~\cite{mccormick1983combinatorial, hoffman1984fast, mccormick1990making, chang1992hierarchical}). These assumptions rarely hold in practice, and specifically, do not apply to any matrix multiplication algorithm we know.

Gottlieb and Neylon's algorithm~\cite{gottlieb2010matrix} sparsifies an $n \times m$ matrix with no assumptions about the input. It does so by using calls to an oracle for the sparsest independent vector problem.

\subsection{Our contribution. }

We obtain three new methods for matrix sparsification, based on Gottlieb and Neylon's~\cite{gottlieb2010matrix} matrix sparsification algorithm. We apply these methods to multiple matrix multiplication algorithms and obtain novel alternative-basis algorithms, often resulting in arithmetic complexity with leading coefficients superior to those known previously (See Table~\ref{tab:ks-compare}, Table~\ref{tab:decomposedalgs}, and Appendix~\ref{sec:results-appendix}).

The first two methods were obtained by the introduction of new solutions to the Sparsest Independent Vector problem, which were then used as oracles for Gottlieb and Neylon's algorithm. As matrix sparsification is known to be NP-Hard, it is no surprise that these methods exhibit exponential worst case complexity. Nevertheless, they perform well in practice on the encoding/decoding matrices of fast matrix multiplication algorithms.

Our third method for matrix sparsification simultaneously minimizes the number of non-singular values in the matrix. This method does not guarantee an optimal solution for matrix sparsification. Nonetheless, it obtains solutions with the same (and, in some cases, better) leading coefficients than the former two methods when applied to many of the fast matrix multiplication algorithms in our corpus, and runs significantly faster than the first two when implemented using Z3~\cite{de2008z3}.
For completeness, we also present the sparsification heuristic used in~\cite{karstadt2017matrix, karstadt2020matrix}.

\subsection{Paper Organization. }

In Section~\ref{sec:perlims}, we recall preliminaries regarding fast matrix multiplication and recursive-bilinear algorithms, followed by a summary of the Alternative Basis technique~\cite{karstadt2017matrix, karstadt2020matrix}. We then present Matrix Sparsification (MS, Problem~\ref{prob:MS}), alongside Gottlieb and Neylon's~\cite{gottlieb2010matrix} algorithm for solving MS by relying on an oracle for Sparsest Independent Vector (SIV, Problem~\ref{prob:SIV}). In Section~\ref{sec:optimal-algs} we present our two algorithms (Algorithms~\ref{SIV1Alg} and \ref{SIV2Alg}) for implementing SIV. In Section~\ref{sec:additional-methos}, we introduce Algorithm~\ref{KSSparseAlg} - the sparsification heuristic of~\cite{karstadt2017matrix, karstadt2020matrix}, and a new efficient heuristic for sparsifying matrices while simultaneously minimizing non-singular values (Algorithm~\ref{BSparseAlg}). In Section~\ref{sec:applications-results} we present the resulting fast matrix multiplication algorithms. Section~\ref{sec:discussion} contains a discussion and plans for future work.

\section{Preliminaries}\label{sec:perlims}

\subsection{Encoding and Decoding matrices. }\label{subsec:fmm}
Fast matrix multiplication algorithms are recursive divide-and-\\conquer algorithms, which utilize a small base case. We use the notation $\strass{n_{0}}{m_{0}}{k_{0}}{t_{0}}$-algorithm to refer to an algorithm multiplying $n_{0}\times m_{0}$ by $m_{0}\times k_{0}$ matrices in its base case, using $t_{0}$ scalar multiplications, where $n_{0},m_{0},k_{0}$ and $t_{0}$ are fixed positive integers.

When multiplying $n\times m$ by $m\times k$ matrix multiplication, the algorithm splits each matrix into blocks (each of size $\frac{n}{n_{0}}\times\frac{m}{m_{0}}$ and $\frac{m}{m_{0}}\times\frac{k}{k_{0}}$, respectively), and works block-wise, according to the base algorithm. Additions and subtractions in the base-case algorithm become block-wise additions and subtractions. Similarly, multiplication by a scalar become multiplication of a block matrix by a scalar. Matrix multiplications in the algorithm are performed via recursion.

Throughout this paper, we refer to an algorithm by its base case. Hence, an $\strass{n}{m}{k}{t}$-algorithm may refer to either the algorithm's base case or the corresponding block recursive algorithm, as obvious from the context.

\begin{fact}\cite{karstadt2017matrix, karstadt2020matrix}\label{fact:encdecmats}
Let $R$ be a ring, and let $f:R^{n}\times R^{m}\to R^{k}$ be a bilinear function that performs $t$ multiplications. There exist $U\in R^{t\times n}, \,V\in R^{t \times m}, \,W\in R^{t \times k}$ such that
\[
	\forall x\in R^{n},\, y\in R^{m},\, \, f\left(x,y\right) = W^{T}\left(\left(U\cdot x\right)\odot \left(V\cdot y\right)\right)
\]
where $\odot$ is the element-wise product (Hadamard product).
\end{fact}

\begin{definition}\cite{karstadt2017matrix, karstadt2020matrix} (Encoding/Decoding matrices). \label{def:encdec}
We refer to the matrix triplet $\encdec{U}{V}{W}$ of a recursive-bilinear algorithm (see Fact \ref{fact:encdecmats}) as its encoding/decoding matrices ($U,\, V$ are the encoding matrices and $W$ is the decoding matrix).
\end{definition}

\begin{notation}~\cite{beniamini2019faster} 
Denote the number of nonzero entries in a matrix by $\nnz{A}$, and the number of non-singleton (i.e., not $\pm 1$) entries in a matrix by $\nns{A}$. Let the number of rows/columns be $nrows\left(A\right)$ and $ncols\left(A\right)$, respectively.
\end{notation}

\begin{remark}\cite{beniamini2019faster} \label{rem:nnz}
The number of linear operations used by a bilinear algorithm is determined by its encoding/decoding matrices. 
The number of arithmetic operations performed by each of the encodings is:
\begin{align*}
	\texttt{OpsU} & = \nnz{U} + \nns{U} - nrows\left(U\right) \\
	\texttt{OpsV} & = \nnz{V} + \nns{V} - nrows\left(V\right)
\end{align*}
The number of operations performed by the decoding is:
\[
	\texttt{OpsW} = \nnz{W} + \nns{W} - ncols\left(W\right)
\]
\end{remark}

\begin{remark}
We assume that none of the rows of the $U,\,V,$ and $W$ matrices is zero. This is because any zero row in $U,\,V$ is equivalent to an identically $0$ multiplicand, and any zero row in $W$ is equivalent to a multiplication that is never used in the output. Hence, such rows can be omitted, resulting in asymptotically faster algorithms.
\end{remark}

\begin{corollary}~\cite{beniamini2019faster} \label{cor:gen-runtime}
Let $ALG$ be an $\strass{n_{0}}{m_{0}}{k_{0}}{t_{0}}$-algorithm that performs $\texttt{OpsU},\, \texttt{OpsV},\, \texttt{OpsW}$ linear operations at the base case and let $n=n_{0}^{l},\,m=m_{0}^{l},\,k=k_{0}^{l}$ ($l\in\mathbb{N}$). The arithmetic complexity of $ALG$ is:
\small{
\begin{align*}
F\left(n,m,k\right) &=\left[1+\frac{\texttt{OpsU}}{t_{0}-n_{0}m_{0}}+\frac{\texttt{OpsV}}{t_{0}-m_{0}k_{0}}+\frac{\texttt{OpsW}}{t_{0}-n_{0}k_{0}}\right]t_{0}^{l}\\
&-\left[\frac{\texttt{OpsU}\cdot nm}{t_{0}-n_{0}m_{0}}+\frac{\texttt{OpsV} \cdot mk}{t_{0}-m_{0}k_{0}}+\frac{\texttt{OpsW}\cdot nk}{t_{0}-n_{0}k_{0}}\right]
\end{align*}
}
\end{corollary}

\begin{definition}\label{def:perm-matrix}
Let $P_{I\times J}$ denote the permutation matrix that exchanges row-order for column-order
of the vectorization of an $I\times J$ matrix.
\end{definition}

\begin{lemma}\label{lem:hop-perm}\cite{hopcroft1973duality} Let $\encdec{U}{V}{W}$
be the encoding/decoding matrices of an $\strass{m}{k}{n}{t}$-algorithm. Then $
\left\langle WP_{n\times m},\,U,\,VP_{n\times k}\right\rangle $
are the encoding/decoding matrices of an $\left\langle n,m,k;t\right\rangle $-algorithm.
\end{lemma}

\begin{remark}\label{rem:perm-operations}
In addition to Lemma~\ref{lem:hop-perm}, Hopcroft and Musinski~\cite{hopcroft1973duality} proved that any $\strass{n}{m}{k}{t}$-algorithm defines algorithms for all permutations of $n,\,m$, and $k$. Note, however, that while the number of non-zero and non-singular entries does not change, it follows from Remark~\ref{rem:nnz} and Corollary~\ref{cor:gen-runtime} that the leading coefficient varies according to the dimensions of the decoding matrix.
\end{remark}

\subsection{Alternative Basis Matrix Multiplication. }\label{subsec:alternative-basis-mult}
\begin{definition}\cite{karstadt2017matrix, karstadt2020matrix}~
Let $R$ be a ring and let $\bmapss$ be automorphisms of $R^{n\cdot m} , \, R^{m\cdot k}, \, R^{n\cdot k}$ (respectively).
We denote a recursive bilinear matrix multiplication algorithm which takes $\phi \left(A\right),\,\psi\left(B\right)$
as inputs and outputs $\upsilon \left(A\cdot B\right)$ using $t$ multiplications by $\strassaltdef{n}{m}{k}{t}$. 
If $n=m=k$ and $\phi=\psi=\upsilon$, we can use the notation $\squarestrass{n}{t}_{\phi}$-algorithm. This notation
extends the $\strass{n}{m}{k}{t}$-algorithm notation, as the latter applies when the three basis transformations are the identity map.
\end{definition}

Given a recursive bilinear, $\strassaltdef{n}{m}{k}{t}$-algorithm {\em ALG}, an alternative basis matrix multiplication operates as follows:

\begin{algorithm}[h]
	\begin{algorithmic}[1]
	\Require{$A\in R^{n\times m}$, $B^{m\times k}$}
	\Ensure{$n\times k$ matrix $C = A\cdot B$}
    	\Function{$Mult$}{$A,B$}
			\State $\tilde{A} = \phi (A)$ \Comment{$R^{n\times m}$ basis transformation}
			\State $\tilde{B} = \psi (B)$ \Comment{$R^{m\times k}$ basis transformation}
			\State $\tilde{C} = ALG (\tilde{A},\tilde{B})$\Comment{$\strassalt{n}{m}{k}{t}{\phi}{\psi}{\upsilon}$-algorithm }
			\State $C = \upsilon^{-1}(\tilde{C})$ \Comment{$R^{n\times k}$ basis transformation}
		\EndFunction
\State \Return{$C$}
	\end{algorithmic}
	\caption{\label{alg:ABS}Alternative Basis Matrix Multiplication Algorithm}
\end{algorithm}

\begin{lemma}\cite{karstadt2017matrix, karstadt2020matrix}\label{cor:alt-to-mul}~Let $R$ be a ring, and let $\bmapss$ be automorphisms of $R^{n\cdot m},\,R^{m\cdot k},\,R^{n\cdot k}$ (respectively). Then $\encdec{U}{V}{W} $
are encoding/decoding matrices of an $\strassaltdef{n}{m}{k}{t}$-algorithm if and only if~$\encdec{U\phi}{V\psi}{W\upsilon^{-T}}$ are encoding/decoding matrices of an $\strass{n}{m}{k}{t}$-algorithm
\end{lemma}

Alternative basis multiplication is fast since the basis transformations are fast and incur an asymptotically negligible overhead:

\begin{claim}\cite{karstadt2017matrix, karstadt2020matrix} \label{cl:basis-runtime}
Let $R$ be a ring, let $\psi :R^{n_{0}\times m_{0}}\to R^{n_{0}\times m_{0}}$ be a linear map, and let $A\in R^{n\times m}$ where $n=n_{0}^{k},\,m=m_{0}^{k}$. The complexity of $\psi\left(A\right)$ is 
\[
F\left(n,m\right)=\frac{q}{n_{0}m_{0}}nm \cdot\log_{n_{0}m_{0}}\left(nm\right)
\]
where $q$ is the number of linear operations performed.
\end{claim}

\subsection{Matrix Sparsification. }

Finding a basis that minimizes the number of additions and subtractions performed by a fast matrix multiplication algorithm is equivalent, by Remark~\ref{rem:nnz}, to the Matrix Sparsification problem:

\begin{problem}\label{prob:MS} Matrix Sparsification Problem \em{(MS)}: 
Let $U$ be an $n\times m$ matrix. The objective is to find an invertible matrix $A$ such that
\[
		A = \underset{A\in GL_n}{\texttt{argmin}}\left(\nnz{AU}\right)
\]
\end{problem}

\begin{remark}\label{rem:matrix-assumptions}
It is traditional to think of the matrices $U,\,V$, and $W$ as ``tall and skinny'', i.e., with $n\geq m$. However, in the area of matrix sparsification, it is traditional to deal with matrices satisfying $n\leq m$ and transformations applied from the \textit{left}. However, since $nnz(AU) = nnz({U^T}{A^T})$, we can simply apply MS to $U^T$ and use $A^T$ as our basis transformation. From now on, we will therefore switch to the convention $n\leq m$ used in matrix sparsification.
\end{remark}

To solve MS, we make use of Gottlieb and Neylon's algorithm~\cite{gottlieb2010matrix}, which solves the matrix sparsification problem for $n\times m$ matrices, by repeatedly invoking an oracle for the Sparsest Independent Vector problem (Problem~\ref{prob:SIV}).

\begin{problem}\label{prob:SIV} Sparsest Independent Vector Problem \em{(SIV)}: 
Let $U\in R^{n\times m}$ ($n\leq m$) and let $\Omega = \left\lbrace \omega_{1},\ldots,\,\omega_{k} \right\rbrace \subset \left[m\right]$. Find a vector $v\in R^{n}$ s.t. $v$ is in the row space of $U$, $v$ is not in the span of $\left\lbrace U_{\omega_{1}},\ldots,\,U_{\omega_{k}} \right\rbrace$, and $v$ has a minimal number of nonzero entries.
\end{problem}

Given a subroutine $SIV\left(U,\Omega\right)$ which returns a pair $\left(v,\,i\right)$, where $v$ is the sparse vector as required by SIV, and $i\in\left[n\right]\setminus\Omega$ is an integer such that the $i$'th row of $U$ can be replaced by $v$ without changing the span of $U$. Then Algorithm~\ref{MSAlg} returns an exact solution for MS~\cite{gottlieb2010matrix}.

\begin{algorithm}
\caption{\label{MSAlg}MS via SIV~\cite{gottlieb2010matrix}}
\begin{algorithmic}[1]
\Procedure{MS(\it{U})}{}
\State $\Omega \gets \emptyset$
\For{$j = 1,...,n$}
    \State $(v_j,i) \gets SIV(U, \Omega)$
    \State Replace $i$'th row of $U$ with $v_j$
    \State $\Omega \gets \Omega \cup \lbrace i \rbrace$
\EndFor
\Return $U$
\EndProcedure
\end{algorithmic}
\end{algorithm}

\section{Optimal Sparsification Methods}\label{sec:optimal-algs}

In this section, we reframe SIV as a problem of finding a maximal subset of columns of the input matrix $U$ according to constraints given by $\Omega$ (see Definition~\ref{def:omega-valid}). We refer to such sets as \OmValid ~sets and show that \OmValid ~sets are tied to sparse independent vectors (Section~\ref{subsec:omvalid-indep}) and that any algorithm which finds an \OmValid ~set of maximal cardinality can be used as an oracle in Algorithm~\ref{MSAlg}. Finally, we show how to find maximal \OmValid ~sets (Section~\ref{subsec:computing-omvalid}), and obtain two algorithms that solve SIV.

Recall that we use the convention that $U \in \mathbb{F}^{n\times m}$ where $n\leq m$ (see Remark~\ref{rem:matrix-assumptions}). Throughout this section, we also assume that $U$ is of full rank $n$ and $\Omega \subsetneq \left[n\right]$.

\begin{notation}
For a set $S$ and an integer $k$, let $\mathcal{C}_{k}\left(S\right)$ denote the set of all subsets of $S$ with $k$ elements.
\end{notation}

\begin{definition}\label{def:omega-valid}
$S\subset \left[m\right]$ is \OmValid ~if there exists $i \notin \Omega$ such that $U_{i,\,S}$ is in the span of $rows \left( U_{\left[n\right] \setminus \left\lbrace i \right\rbrace ,\,S}\right)$.

Formally, a set $S\subset \left[m\right]$ is \OmValid ~if exists $\lambda \in \mathbb{F}^{n}$ with $supp\left(\lambda\right) \not\subset \Omega$ s.t. $\lambda^{T} U_{:,\,S}=0$ (where $supp\left(\lambda\right)=\left\lbrace i : \lambda_{i} \neq 0 \right\rbrace$).
\end{definition}

\begin{notation}
Given an \OmValid ~set $S$, we will refer to a vector $\lambda \in \mathbb{F}^{n}$ with $supp\left(\lambda\right) \not\subset \Omega$ s.t. $\lambda^{T} U_{:,\,S}=0$ as an \OmValidator ~of $S$.
\end{notation}

Next, we provide a definition for vectors which are candidates for a solution of SIV:
\begin{definition}\label{def:indep-of-omega}
A vector $v$ in the row space of $U$ is called $\Omega$-independent if $v$ is not in the row space of $U_{\Omega,\,:}$.
\end{definition}
Note that any solution to SIV (Problem~\ref{prob:SIV}) is, by definition, an optimally sparse $\Omega$-independent vector.

\begin{remark}
Note that given a set $S\subset \left[m\right]$, it is possible to verify whether $S$ is \OmValid and find an appropriate \OmValidator ~ for it in cubic time (e.g., via Gaussian elimination).
\end{remark}

\subsection{Sparse Independent Vectors and maximal \OmValid ~sets. } \label{subsec:omvalid-indep}

The crux of our algorithms lies in the idea of finding an \OmValid ~set of maximal cardinality and using it to compute a solution for SIV, which can then be used by Algorithm~\ref{MSAlg}. The connection between \OmValid ~sets and $\Omega$-independent vectors is given by the following lemmas:

\begin{lemma} \label{lem:indep-equal-valid-zeros}
Let $v\in \mathbb{F}^{n}$ be an \OmIndep ~vector. Then the set $S=\left\lbrace j : v_{j}=0\right\rbrace$ is an \OmValid ~set of size $zeros\left(v\right)$.
\begin{proof}
By Definition~\ref{def:indep-of-omega}, there exists a vector $\lambda\in \mathbb{F}^{n}$ s.t. $v=\sum_{i=1}^{n} \lambda_{i} U_{i}$ (i.e., $v=\lambda^{T}U$) and $\lambda_{i_0}\neq 0$ for some $i_{0}\notin \Omega$ (hence $supp\left(\lambda\right)\not\subset \Omega$). Thus, $\lambda$ is an \OmValidator ~of $S$, and therefore, $S$ is \OmValid .
\end{proof}
\end{lemma}

\begin{lemma}\label{lem:valid-at-least-zeros}
Let $S\subset\left[m\right]$ be an \OmValid ~set and let $\lambda \in\mathbb{F}^{n}$ an \OmValidator ~of $S$. Then $v=\lambda^{T} U$ is an \OmIndep ~vector with at least $\left\lvert S \right\rvert$ zero entries.
\begin{proof}
Since $S$ is valid, there exists $\lambda \in \mathbb{F}^{n}$ s.t. $\lambda^{T} U_{:,S} = 0$ and $supp\left( \lambda \right)\not\subset \Omega$. Denote $v=\lambda^{T} U$. By definition, $v$ has at least $\left\lvert S \right\rvert$ zero entries since $\forall i\in S\, v_{i}=\left(\lambda^{T} U_{:,\,S}\right)_{i}=0$.
Next we show that $v$ is \OmIndep . Note that, $v=\lambda^{T} U = \sum_{i=1}^{n} \lambda_{i} U_{i,\,:}$ is in the row space of $U$ since it is a linear combination of the rows of $U$. Furthermore, since $supp\left(\lambda\right)\not\subset\Omega$, there exists $i_{0}\notin\Omega$ s.t. $\lambda_{i_0} \neq 0$. Therefore, $v$ is not in the row span of $U_{\Omega,\,:}$ since we assume (Remark~\ref{rem:matrix-assumptions}) that all rows of $U$ are linearly independent.
Hence, $v=\lambda^{T}U$ is an \OmIndep ~vector with at least $\left\lvert S \right\rvert$ zero entries.
\end{proof}
\end{lemma}

\begin{corollary}\label{cor:maximal-valid-equal-zeros}
Let $M\subset\left[m\right]$ be a maximal \OmValid ~set (i.e., $M$ is not a subset of any other \OmValid ~set), and let $v\in\mathbb{F}^{n}$ be an \OmIndep ~vector s.t. $\forall i\in M\, v_{i}=0$. Then $\forall j\notin M,\,v_{j}\neq 0$.
\begin{proof}
Denote the set of indices of zero entries of $v$ by $M' = \left\lbrace j : v_{j}=0\right\rbrace$. Since $v$ is \OmIndep , Lemma~\ref{lem:indep-equal-valid-zeros} yields that $M'$ is valid. Hence, by maximality of $M$, $M=M'$ and $\left\lvert M \right\rvert = zeros\left(v\right)$. Therefore, $\forall i \in \left[m\right] \, v_{i}=0$ if, and only if, $i\in M$.
\end{proof}
\end{corollary}

\begin{corollary}\label{cor:maximal-valid-equal-max-sparse}
Let $M\subset\left[m\right]$ be a maximal \OmValid ~set and let $\lambda\in\mathbb{F}^{n}$ be an \OmValidator ~of $M$. Then $v=\lambda^{T}U$ is an \OmIndep ~\\vector with exactly $\left\lvert M \right\rvert$ zero entries.
\begin{proof}
Follows directly from Lemma~\ref{lem:valid-at-least-zeros} and Corollary~\ref{cor:maximal-valid-equal-zeros}
\end{proof}
\end{corollary}

The final two claims will show how \OmValidity ~can serve as an oracle for Algorithm~\ref{MSAlg}. Recall that Algorithm~\ref{MSAlg} uses an oracle which returns a pair $\left(v,i\right)$, where $v$ is an optimally sparse \OmIndep ~vector, and replacing the $i$'th row of $U$ with $v$ does not change the row span of $U$. The next claim shows that a maximally sparse \OmIndep ~vector is equivalent to an \OmValid ~set of maximal cardinality.

\begin{claim}\label{cl:solve-siv-iff-max-card-valid}
An \OmIndep ~vector $v\in\mathbb{F}^{m}$ is optimally sparse if, and only if, $M=\left\lbrace i\,:\,v_{i}=0\right\rbrace$ is an \OmValid ~set of maximal cardinality.
\begin{proof}
First, assume that $v\in\mathbb{F}^{m}$ is a maximally sparse \OmIndep ~vector (i.e., for any \OmIndep ~vector $u$, $zeros\left(u\right)\leq zeros\left(v\right)$). From Lemma~\ref{lem:indep-equal-valid-zeros}, we know that $M$ is \OmValid . Lemma~\ref{lem:valid-at-least-zeros} shows that if there exists an \OmValid ~set $S$ s.t. $\left\lvert M \right\rvert < \left\lvert S \right\rvert$, then there also exists an \OmIndep ~vector $u\in\mathbb{F}^{m}$ s.t. $zeros\left(u\right) \geq \left\lvert S \right\rvert > zeros\left(v\right)$. This contradicts $v$ being a maximally sparse \OmIndep ~\\vector.

Now, assume that $M$ is an \OmValid ~set of maximal cardinality (i.e., for any \OmValid ~set $S$, $\left\lvert S \right\rvert \leq \left\lvert M \right\rvert$) and let $\lambda_{M}$ be an \OmValidator ~of $M$. By Corollary~\ref{cor:maximal-valid-equal-max-sparse}, $v_{M}=\lambda_{M}^{T}U$ is an \OmIndep ~vector with exactly $\left\lvert M \right\rvert$ zero entries. Assume by contradiction that exists $u\in\mathbb{F}^{m}$ with $z > \left\lvert M \right\rvert$ zero entries, then by Lemma~\ref{lem:indep-equal-valid-zeros}, there is an \OmValid ~set $S$ s.t. $\left\lvert M \right\rvert < \left\lvert S \right\rvert$, in contradiction to $M$ being an \OmValid ~set of maximal cardinality. Therefore, $v_{M}=\lambda_{M} U$ is a maximally sparse \OmIndep ~vector.
\end{proof}
\end{claim}

The following claim shows that given an \OmValid ~set, $S$, and its corresponding \OmIndep ~vector $v$ (as in Lemma~\ref{lem:valid-at-least-zeros}), the support of the \OmValidator ~of $S$ can be used to find an index $i$ s.t. the $i$'th row of $U$ can be replaced with $v$ without changing the row span of $U$.

\begin{claim}\label{cl:omvalid-correct-index}
Let $S$ be an \OmValid ~set, let $\lambda$ be an \OmValidator ~of $S$, and let $v=\lambda^{T} U$. 
Then for any $i\in supp\left(\lambda\right)\setminus \Omega$, replacing row $i$ of $U$ with $v$ does not the change row span of $U$. That is:
\[
span\left(rows\left(U\right)\right)=span\left(rows\left(U_{\left[n\right]\setminus \left\lbrace i \right\rbrace,:}\right)\cup\left\lbrace v\right\rbrace\right)
\]

\begin{proof}
Fix $i_{0}\in supp\left(\lambda\right)\setminus \Omega$. Since $v$ is a linear combination of rows of $U$ and $\lambda_{i_{0}}\neq 0$, $u\in span\left(rows\left(U_{\left[n\right]\setminus \left\lbrace i_{0} \right\rbrace,:}\right)\cup\left\lbrace v\right\rbrace\right)$, for any $u\in span\left(rows\left(U\right)\right)$.  Now, let $\alpha\in \mathbb{F}^{n}$ be the vector $\alpha_{j}=-\lambda_{j}$ (for $j\neq i_{0}$) and $\alpha_{i_{0}}=0$. Then $w=\alpha^{T}U\in span\left(rows\left(U_{\left[n\right]\setminus \left\lbrace i_{0} \right\rbrace,:}\right)\right)$, therefore $w+v=\lambda_{i_{0}}U_{i_{0},:}\in span\left(rows\left(U_{\left[n\right]\setminus \left\lbrace i_{0} \right\rbrace,:}\right)\cup\left\lbrace v\right\rbrace\right)$. Hence, $span\left(rows\left(U\right)\right)=$ \\ $span\left(rows\left(U_{\left[n\right]\setminus \left\lbrace i \right\rbrace,:}\right)\cup\left\lbrace v\right\rbrace\right)$.
\end{proof}
\end{claim}

Therefore, any algorithm which finds an \OmValid ~set of maximal cardinality is an oracle for Algorithm~\ref{MSAlg}.

\subsection{Computing maximal \OmValid ~sets. }\label{subsec:computing-omvalid}

Given a maximal \OmValid ~set, we now have the tools to compute optimally sparse \OmIndep ~vectors. As the next stage, we show how to compute a maximal \OmValid ~set $M$ using a small subset of columns $S\subset M$. The key intuition here is that if $\lambda\in\mathbb{F}^{n}$ is an \OmValidator ~of $S$, then $\lambda$ is orthogonal to all columns indexed by $S$ (since $\lambda^{T} U(:, S) = 0$), and any linear combinations of columns of $S$. This leads to the following extension of sets:

\begin{definition}\label{def:max-extension}
Let $S\subset\left[m\right]$. We define the extension of $S$, $E\left(S\right)$, to be the largest set $E\subset\left[m\right]$ s.t. $span\left(col\left(U_{:,\,S}\right)\right)=span\left(col\left(U_{:,\,E}\right)\right)$.
\end{definition}

\begin{lemma} \label{lem:valid-extension}
Let $S\subset \left[m\right]$. Then $S$ is \OmValid ~if, and only if, $E\left(S\right)$ is \OmValid .
\begin{proof}
Assume $E\left(S\right)$ is \OmValid . By definition of \OmValidity , exists a vector $\lambda\in\mathbb{F}^{n}$ s.t. $supp\left(\lambda\right)\not\subset\Omega$ and $\lambda^{T} U_{:,E\left(S\right)}=0$. Since $S \subset E\left(S\right)$, $\lambda^{T} U_{:,S}=0$, therefore, $S$ is valid.

Let $S\subset\left[m\right]$ be an \OmValid ~set, and let $\lambda\in \mathbb{F}^{n}$ with $supp\left(\lambda\right) \not\subset\Omega$ s.t. $\lambda^{T} U_{:,\, S}=0$. Since $col\left(U\left(:,\, E\left(S\right)\right)\right) = col\left(U\left(:,\, S\right)\right)$, all columns indexed by $E\left(S\right)$ are linear combinations of the columns indexed by $S$. Since $\lambda$ is orthogonal to all columns of $U$ indexed by $S$, it is also orthogonal to all their
linear combinations. Therefore, $\lambda^{T} U_{:,\, E\left(S\right)} = 0$. Hence $E\left(S\right)$ is valid.
\end{proof}
\end{lemma}

Next we show that the search for a maximal \OmValid ~set can be reduced to the search over maximal extensions of sets of size $n-1$.

\begin{remark}\label{rem:max-rank-of-valid}
Note that $rank\left(U_{:,\,S}\right)\leq n-1$ for any \OmValid ~set $S$. This is due to the fact that if $rank\left(U_{:,\,S}\right)=n$ then $\lambda^{T} U_{:,\,S} = 0$ implies that $\lambda=0$ since the rows of $U$ are linearly independent.
\end{remark}

\begin{lemma}\label{lem:compute-valid}
Let $S$ be an \OmValid ~set and let $\lambda\in\mathbb{F}^{n}$ be an \OmValidator ~of $S$. Then 
\[
	E\left(S\right) \subset \left\lbrace i\,:\, \left(\lambda^{T}U\right)_{i} = 0 \right\rbrace
\]
\begin{proof}
Let $D = \left\lbrace i\,:\, \left(\lambda^{T}U\right)_{i} = 0 \right\rbrace$. By Definition~\ref{def:max-extension}, columns indexed by $E\left(S\right)$ are linear combinations of the columns indexed by $S$ and $\lambda$ is orthogonal to all columns of $U_{:,\,S}$ (and their linear combinations). Hence, $\lambda^{T} U_{:,\, E\left(S\right)} = 0 $ and $E\left(S\right) \subset D$.
\end{proof}
\end{lemma}

\begin{lemma}\label{lem:extension-is-maximal}
Let $S$ be an \OmValid ~set s.t. $rank\left(U_{:,\,S}\right)=n-1$, and let $D$ be an \OmValid ~set s.t. $S\subset D$. Then $D\subset E\left(S\right)$.
\begin{proof}
Since $S\subset D$, $n-1 = rank\left(U_{:,\,S}\right) \leq rank\left(U_{:,\,D}\right)$. However, from Remark~\ref{rem:max-rank-of-valid}, we know that $rank\left(U_{:,\,D}\right)\leq n-1$, therefore, $span\left(col\left( U_{:,\,S}\right)\right)=span\left(col\left(U_{:,\,D}\right)\right)$. Hence, by definition, $D\subset E\left(S\right)$.
\end{proof}
\end{lemma}

\begin{corollary}\label{cor:compute-max-valid}
Let $S$ be an \OmValid ~set s.t. $rank\left(U_{:,\,S}\right)=n-1$, and let $\lambda\in\mathbb{F}^{n}$ be an \OmValidator ~of $S$. Then 
\[
	E\left(S\right) = \left\lbrace i\,:\, \left(\lambda^{T}U\right)_{i} = 0 \right\rbrace
\]
\begin{proof}
This is a direct result of Lemma~\ref{lem:compute-valid} and Lemma~\ref{lem:extension-is-maximal}.
\end{proof}
\end{corollary}

Note that Corollary~\ref{cor:compute-max-valid} gives us the tools to quickly compute the extension of any \OmValid ~set $S$ such that $rank\left(U_{:,\,S}\right)=n-1$. Next we prove that any maximal \OmValid ~set is an extension of an \OmValid ~set of $n-1$ linearly independent columns of $U$:

\begin{claim} \label{claim:rank-of-maximal-cardinality-valid-set}
Let $S\subset \left[m\right]$ be a maximal \OmValid ~set, then 
\[
rank\left(U_{:,\,S}\right) = n-1
\]
\begin{proof}
Let $S\subset \left[m\right]$ be a maximal \OmValid ~set, and let $i_{0}\notin \Omega$ such that $U_{i_{0},\,S}\in span\left(rows\left(U_{\left[n\right]\setminus i_{0},\,S}\right)\right)$ (such $i_{0}$ exists by definition of an \OmValid ~set). Suppose, by contradiction, that $rank\left(U_{:,\,S}\right)=n-r$ for some $r > 1$.

Note that since $U_{i_{0},\,S}$ is in the row span $U_{\left[n\right]\setminus i_{0},\,S}$, $rank\left(U_{:,\,S}\right)=rank\left(U_{\left[n\right]\setminus i_{0},\,S}\right)=n-r$. Therefore, exists $S_{0}\subset S$ s.t. $|S|=n-r$ and $rank\left(U_{:,\,S_{0}}\right)=n-r$.
 
Let $Q\subset\left[m\right]\setminus S$ s.t. $\left\lvert Q \right\rvert = r-1$, $rank\left(U_{\left[n\right]\setminus i_{0},\,Q}\right)=r-1$, and each column indexed by $Q$ is not in the column span of $U_{\left[n\right]\setminus i_{0},\, S}$. Such $Q$ exists because the matrix $U_{\left[n\right]\setminus \left\lbrace i_{0} \right\rbrace,\, :}$ has full rank $n-1$ (since $U$ is of full row rank $n$).

Since the matrix $U_{\left[n\right]\setminus \left\lbrace i_{0} \right\rbrace,\, S_{0}\cup Q}$ is a square $n-1\times n-1$ matrix of full rank, $U_{i_{0}, S_{0}\cup Q}$ is in the span of $row\left(U_{\left[n\right]\setminus \left\lbrace i_{0} \right\rbrace,\, S_{0}\cup Q}\right)$. Therefore, $S_{0}\cup Q$ is an \OmValid ~set.

By Lemma \ref{lem:valid-extension}, the extension of $S_{0}\cup Q$ is also valid. Furthermore, $S\cup Q\subset E\left(S_{0}\cup Q\right)$ because we have chosen $S_{0}$ s.t. it spans the same column space as $S$. However, by construction of $Q$, we know that $S\cap Q=\emptyset$, meaning that $\left\lvert E\left(S_{0}\cup Q\right)\right\rvert \geq \left\lvert S\cup Q \right\rvert > \left\lvert S \right\rvert$. This in contradiction to maximality of $S$.
\end{proof}
\end{claim}

\begin{corollary}\label{cor:max-valid-is-extension}
Let $S\subset \left[m\right]$ be a maximal \OmValid ~set and let $C\subset S$ s.t. $rank\left(U_{:,\,C}\right) = n-1$. Then $S=E\left(C\right)$.
\begin{proof}
$C\subset S$, Therefore, $span\left(col\left(U_{:,\,C}\right)\right)\subset span\left(col\left(U_{:,\,S}\right)\right)$. Because $S$ is maximal, Claim~\ref{claim:rank-of-maximal-cardinality-valid-set} shows that $rank\left(U_{:,\,S}\right) = n-1$. We have, by  rank equality, that $span\left(col\left(U_{:,\,C}\right)\right) = span\left(col\left(U_{:,\,S}\right)\right)$. By definition, $E\left(C\right)$ is the maximal set $E$ s.t. $span\left(col\left(U_{:,\,C}\right)\right)\subset span\left(col\left(U_{:,\,E}\right)\right)$, therefore, $S\subset E\left(C\right)$. However, by maximality of $S$, we have $S=E\left(C\right)$.
\end{proof}
\end{corollary}

\begin{corollary}\label{cor:generating-set-size}
Let $S\subset\left[m\right]$ be a maximal \OmValid ~set, then exist $C\in \mathcal{C}_{n-1}\left(\left[m\right]\right)$ s.t. $S=E\left(C\right)$.
\begin{proof}
This is a direct result of Corollary~\ref{cor:max-valid-is-extension}
\end{proof}
\end{corollary}

\subsection{First algorithm for SIV. }

Our first algorithm performs an exhaustive search over all maximal \OmValid ~sets in order find one with maximal cardinality. This is a result of the observation given by Claim~\ref{cl:solve-siv-iff-max-card-valid}, which states that any solution to SIV is tied to an \OmValid ~set of maximal cardinality (and vice versa).
The search is done using by combining Corollary~\ref{cor:generating-set-size}, which states that any maximal \OmValid ~set is the extension of an \OmValid ~set of $n-1$ independent columns, and Corollary~\ref{cor:compute-max-valid}, which provides a method to compute said extension.

\begin{algorithm}[h]
\caption{\label{SIV1Alg}Sparsest Independent Vector (1)}
\begin{algorithmic}[1]
\Procedure{$SIV(\it{U}, \it{\Omega})$}{}
\State sparsity $\gets 0$
\State sparsest $\gets$ null
\State $i \gets $ null
\For{$C \in \mathcal{C}_{n - 1}(\lbrace 1,...,m \rbrace)$}
    \If{$rank\left(U_{:,C}\right)< n-1$ or $C$ is not \OmValid }
    	\State continue
    \EndIf
    
	\State $\lambda \gets$ \OmValidator ~of $C$
    \State $v \gets \lambda^{T} U$
    \State $E \gets \left\lbrace i\,: v_{i} = 0 \right\rbrace$
    \If{$\left\lvert E \right\rvert >$ sparsity }
        \State sparsity $\gets \left\lvert E \right\rvert$
        \State sparsest $\gets v$
        \State $i \gets $ any element of $supp\left(\lambda\right) \setminus \Omega$
    \EndIf   
\EndFor
\Return $\left(v,\,i\right)$
\EndProcedure
\end{algorithmic}
\end{algorithm}

\begin{lemma}\label{lem:SIV1Alg-iterates-all-maximal}
Algorithm~\ref{SIV1Alg} iterates over all maximal \OmValid ~sets.
\begin{proof}
By Corollary~\ref{cor:generating-set-size}, for any maximal \OmValid ~set $E$, there exist an \OmValid ~set $C\in \mathcal{C}_{n-1}\left(\left[m\right]\right)$ s.t. $rank\left(U_{:,\,C}\right)=n-1$ and $E$ is the extension of $C$. Therefore, the algorithm iterates over all \OmValid ~sets $C\in \mathcal{C}_{n-1}\left(\left[m\right]\right)$ s.t. $rank\left(U_{:,\,C}\right)=n-1$. Furthermore, by Corollary~\ref{cor:compute-max-valid}, if $rank\left(U_{:,\,C}\right)=n-1$ and $\lambda$ is an \OmValidator ~of $C$ then $E\left(C\right)=\left\lbrace i\,: \left(\lambda^{T}U\right)_{i} = 0 \right\rbrace$. The algorithm performs this computation at lines 11-13. Hence, the algorithm iterates over all \OmValid ~sets.
\end{proof}
\end{lemma}

\begin{theorem} Algorithm~\ref{SIV1Alg} produces an optimal solution to SIV, and is an oracle for Algorithm~\ref{MSAlg}.

\begin{proof} 
By Lemma~\ref{lem:SIV1Alg-iterates-all-maximal}, Algorithm~\ref{SIV1Alg} iterates over all maximal \OmValid ~sets. Lines 14-17 check whether a given \OmValid ~set has greater cardinality than any previously found maximal \OmValid ~set and if it does, the algorithm choose this set as a working solution. Hence, at the end of the algorithm, the chosen vector $v$ correlates to a maximal cardinality \OmValid ~set. By Claim~\ref{cl:solve-siv-iff-max-card-valid}, $v$ is an optimal solution to SIV (a maximally sparse \OmIndep ~vector) if, and only if, the set $E = \left\lbrace i\,: v_{i} = 0 \right\rbrace$ is an \OmValid ~set of maximal cardinality. Therefore, the vector chosen at the end of the algorithm is a maximally sparse \OmIndep ~vector. Finally, by Claim~\ref{cl:omvalid-correct-index}, the pair $\left(v,i\right)$ serves as the oracle for SIV required by Algorithm~\ref{MSAlg}.
\end{proof}
\end{theorem}

\subsection{Implementation of our first optimal algorithm. }\label{subsec:nathan-implementation}
In order for Algorithm~\ref{SIV1Alg} to perform well, we have added a blacklist to the algorithm's operation. Since the maximal \OmValid ~sets are generated by computing the extension (Definition~\ref{def:max-extension}) of $n-1$ independent columns, once a given \OmValid ~set is found, we wish to blacklist all of its subsets of size $n-1$ since we need not revisit that extension. However, in addition to memory costs, looking up an element in the blacklist incurs a significant overhead as the blacklist grows. To address this problem, rather than storing all subsets $\mathcal{C}_{n-1}\left(S\right)$ of a given set $S$, we store $S$ itself in the blacklist, in which case $C$ is not blacklisted if $\forall B\in blacklist$ $C\not\subset B$. Despite this measure, in some cases the blacklist still grew too large, so we and imposed a limit on the maximum size of the blacklist, storing only the $M$ largest sets found so far.

\subsection{Second algorithm for SIV. }

While our first algorithm performs well in many cases, we have found that it performs poorly when the largest \OmValid ~set is very large. In such cases the algorithm quickly finds the correct solution, but then continues its exhaustive search for a very long time. Our second algorithm is slightly simpler and avoids this inefficiency by using a top-down approach, searching for \OmValid ~sets in descending order of cardinality to find an \OmValid ~set of maximal cardinality. Just like our first algorithm, it relies on the observation of Claim~\ref{cl:solve-siv-iff-max-card-valid}, which ties any solution of SIV (maximally sparse, \OmIndep ~vector) to an \OmValid ~set of maximal cardinality. 

\begin{algorithm}
\caption{\label{SIV2Alg}Sparsest Independent Vector (2)}
\begin{algorithmic}[1]
\Procedure{$SIV\left(\it{U}, \it{\Omega}\right)$}{}
\For{$z = m - 1,...,n - 1$}
    \For{$C \in \mathcal{C}_{z}\left(\left[m\right]\right)$}
    	\If{$rank\left(U_{:,C}\right)=n-1$ and $C$ is \OmValid }
			\State $\lambda \gets$ \OmValidator ~of $C$
		    \State $v \gets \lambda^{T} U$
	        \State $i \gets $ any element of $supp\left(\lambda\right) \setminus \Omega$
            \State \Return $\left(v,\,i\right)$
    \EndIf   
    \EndFor
\EndFor
\EndProcedure
\end{algorithmic}
\end{algorithm}

To prove the correctness of our Algorithm~\ref{SIV2Alg}, we use the following lemma, which provides bounds on the size of a maximal \OmValid ~set.

\begin{lemma}\label{lem:size-of-maximal-cardinality-valid}
Let $S\subset\left[m\right]$ be a maximal \OmValid ~set, then $n-1\leq\left\lvert S \right\rvert\leq m-1$.
\begin{proof}
First, we show that $\left\lvert S \right\rvert < m$. Assume, by contradiction, that $\left\lvert S \right\rvert = m$ and let $\lambda\in\mathbb{F}^{n}$ be an \OmValidator ~of $S$. Then $\lambda^{T}U=0$, which means that $\sum_{i\in\left[n\right]} \lambda_{i} U_{i,:} = 0$, in contradiction to $U$ having full row rank $n$. Hence, $\left\lvert S \right\rvert \leq m-1$.

Next, by Claim~\ref{claim:rank-of-maximal-cardinality-valid-set}, since $S$ is a maximal \OmValid ~set, its rank is$n-1$, therefore, $n-1 \leq \left\lvert S \right\rvert$. Hence $n-1 \leq \left\lvert S \right\rvert \leq m-1$.
\end{proof}
\end{lemma}

\begin{theorem}
Algorithm~\ref{SIV1Alg} produces an optimal solution to SIV, and is an oracle for Algorithm~\ref{MSAlg}.
\end{theorem}

\begin{proof} 
Claim~\ref{cl:solve-siv-iff-max-card-valid} states that $v\in \mathbb{F}^{m}$ is a solution to SIV (an optimally sparse, \OmIndep ~vector) if and only if $S=\left\lbrace i\,:\, v_{i}=0\right\rbrace$ is \OmValid . The algorithm iterates all subsets of $\left[m\right]$ in descending order of cardinality. Therefore, the first \OmValid ~set found is an \OmValid ~set of maximal cardinality. Furthermore, Lemma~\ref{lem:size-of-maximal-cardinality-valid} states that any maximal \OmValid ~set is of size $n-1\leq z \leq m-1$, hence, the algorithm iterates all candidates $S\subset \left[m\right]$ that could be \OmValid ~sets of maximal cardinality. Therefore, Algorithm~\ref{SIV2Alg} returns a sparsest \OmIndep ~vector. Finally, by Claim~\ref{cl:omvalid-correct-index}, the pair $\left(v,i\right)$ serves as the oracle for SIV required by Algorithm~\ref{MSAlg}.
\end{proof}
\section{Additional Sparsification Methods}\label{sec:additional-methos}

\subsection{Sparsification via subset of rows.}\label{subsec:ks-heuristic}
The alternative bases presented in Karstadt and Schwartz's~\cite{karstadt2017matrix, karstadt2020matrix} paper were found using a straightforward heuristic of iterating over all sets of $n$ linearly independent rows of an $n\times m$ matrix of full rank (where $n\geq m$). This heuristic was based on the observation that using the columns of the original matrix for sparsification ensures that the sparsified matrix contains $n$ rows, each with only a single non-zero entry.

\begin{algorithm}[h]
\caption{\label{KSSparseAlg}Row basis sparsification~\cite{karstadt2017matrix, karstadt2020matrix}}
\begin{algorithmic}[1]
\Procedure{KS-Sparsification(\it{U})}{}
\State sparsity $\gets nnz\left(U\right)$
\State basis $\gets I_{n}$
\For{$C \in \mathcal{C}_{m}\left(n\right)$}
    \If{$U_{:,\,C}$ is of full rank}
        \State sparsifier $\gets U_{:,\,C}^{-1}$
        \If{$\texttt{nnz}\left(sparsifier \cdot U \right) < sparsity$}
            \State sparsity $\gets \texttt{nnz}\left(sparsifier\cdot U  \right)$
            \State basis $\gets$ sparsifier
        \EndIf
    \EndIf
\EndFor
\State \Return basis
\EndProcedure
\end{algorithmic}
\end{algorithm}

While this method is inefficient, requiring $\left(\begin{smallmatrix} m \\ n \end{smallmatrix}\right)$ passes, it finds sparsifications which significantly improve the leading coefficients of multiple algorithms. The refinement of this method led to the development of Algorithm~\ref{SIV1Alg}. It is therefore presented here for completeness.

\subsection{Greedy sparsification.}\label{subsec:gal-heuristic}
A second heuristic for matrix sparsification, inspired by Gottlieb and Neylon's algorithm (Algorithm~\ref{MSAlg}), employs an even simpler greedy approach.

Recall that for a given $n \times m$ matrix $U$ ($n \leq m$), we seek an $n \times n$ matrix $A$ which minimizes $\nnz{AU} + \nns{AU}$. For this purpose, rather than searching for the entire invertible matrix $A$ achieving this objective, we could instead search for each row of $A$ individually. Concretely, we iteratively compose the matrix $A$ row-wise; where at each step $i$, we obtain the sparsest row vector $v_i$ such that $v_{i}$ is independent of $\{v_1,\dots,v_{i-1}\}$ and minimizes $\nnz{vU} + \nns{vU}$. This yields the following algorithm:

\begin{algorithm}[h]
	\caption{\label{BSparseAlg} Greedy Sparsification}
	\begin{algorithmic}[1]
		\Procedure{$Greedy-Sparsification(U)$}{}
		\State $A \gets \emptyset$
		\For{$i=1,\ldots,\,n$}
		\State {\small $v \gets \underset{\substack{v \in \mathbb{F}^m\\ rk(\{v_1,\dots,v_{i-1},v\}) = i}}{\texttt{argmin}}\left(\texttt{nnz}\left(v^{T}U\right) + \texttt{nns}\left(v^{T}U\right)\right)$}
		\State $A_{i,:} \gets v^{T}$
		\EndFor
		\State \Return $A$
		\EndProcedure
	\end{algorithmic}
\end{algorithm}

In order to implement the subroutine for finding each row vector $v_i$, we encoded the objective as a MaxSAT instance and used Z3~\cite{de2008z3}, an SMT Theorem Prover, to find the optimal solution. Our MaxSAT instance employs two types of ``soft'' constraints: one which penalizes non-zero entries, and another which penalizes non-singleton entries. Therefore, optimal solutions will minimize the sum of non-zero and non-singleton entries, thereby minimizing the associated arithmetic complexity (Remark~\ref{rem:nnz}).


This algorithm, while not proven to be optimal, has the advantage of considering both non-zeros and non-singletons, and can therefore produce decompositions resulting in a lower arithmetic complexity than the optimal algorithms (Algorithms~\ref{SIV1Alg}, \ref{SIV2Alg}). For a summary of these results, see Table~\ref{tab:decomposedalgs}.
\begin{table*}[ht]
	\centering
	\caption {Alternative Basis Algorithms} \label{tab:decomposedalgs}
	\small{
	\begin{tabular}{ |l||c||c|c||c|c||c| }
		\hline
		\multirow{3}{*}{\ \ \ Algorithm} &
		\multirow{3}{*}{Leading Monomial} &
		\multicolumn{2}{c||}{Arithmetic Operations} &
		\multicolumn{2}{c||}{Leading Coefficient} &
		\multirow{3}{*}{Improvement}\\
		\cline{3-6}
		& & Original & Here & Original & Here & \\
		\hline
		&&&&&&\\[-1.05em]
		$\langle 2,2,2;7 \rangle$~\cite{strassen1969gaussian} &$n^{\log_{2} 7} \approx n^{2.80735}$& 18 & 12 & 7 & 5 & 28.57\%\\
		\cmidrule[1.1pt]{1-7}
		&&&&&&\\[-1.05em]
		$\langle 3,2,2;11 \rangle$~\cite{benson2015framework} & $n^{\log_{12} 11^{3}} \approx n^{2.89495}$ & 22 & 18 & 5.06 & 4.26 & 15.82\%\\
		\hline
		&&&&&&\\[-1.05em]
		$\langle 2,3,2;11 \rangle$~\cite{madan2015matrix} & $n^{\log_{12} 11^{3}} \approx n^{2.89495}$ & 22 & 18 & 4.71 & 3.91 & 16.97\%\\
		\cmidrule[1.1pt]{1-7}
		&&&&&&\\[-1.05em]
		$\langle 4,2,2;14 \rangle$~\cite{benson2015framework} & $n^{\log_{16} 14^{3}} \approx n^{2.85551}$ & 48 & 28 & 8.33 & 5.27 & 36.8\%\\
		\cmidrule[1.1pt]{1-7}
		&&&&&&\\[-1.05em]
		$\langle 3,2,3;15 \rangle$~\cite{hopcroft1971minimizing} &$n^{\log_{18} 15^{3}} \approx n^{2.81076}$& 55 & 39 & 8.28 & 6.17 & 25.5\%\\
		\hline
		&&&&&&\\[-1.05em]
		$\langle 3,2,3;15 \rangle$~\cite{benson2015framework} &$n^{\log_{18} 15^{3}} \approx n^{2.81076}$& 64 & 39 & 9.61 & 6.17 & 35.84\%\\
		\cmidrule[1.1pt]{1-7}
		&&&&&&\\[-1.05em]
		$\langle 5,2,2;18 \rangle$~\cite{benson2015framework} & $n^{\log_{20} 18^{3}} \approx n^{2.89449}$ & 53 & 32 & 6.98 & 4.46 & 36.06\%\\
		\cmidrule[1.1pt]{1-7}
		&&&&&&\\[-1.05em]
		$\langle 4,2,3;20 \rangle$~\cite{smirnov2013bilinear} &$n^{\log_{24} 20^{3}} \approx n^{2.82789}$& 78 & 51 & 8.9 & 5.88 & 33.96\%\\
		\hline
		&&&&&&\\[-1.05em]
		$\langle 4,2,3;20 \rangle$~\cite{benson2015framework} &$n^{\log_{24} 20^{3}} \approx n^{2.82789}$& 82 & 51 & 9.19 & 5.88 & 36.01\%\\
		\hline
		&&&&&&\\[-1.05em]
		$\langle 4,2,3;20 \rangle$~\cite{benson2015framework} &$n^{\log_{24} 20^{3}} \approx n^{2.82789}$& 86 & 54 & 9.38 & 6.12 & 34.77\%\\
		\hline
		&&&&&&\\[-1.05em]
		$\langle 4,2,3;20 \rangle$~\cite{benson2015framework} &$n^{\log_{24} 20^{3}} \approx n^{2.82789}$& 104 & 56 & 11.38 & 6.38 & 43.9\%\\
		\hline
		&&&&&&\\[-1.05em]
		$\langle 2,3,4;20 \rangle$~\cite{benson2015framework} &$n^{\log_{24} 20^{3}} \approx n^{2.82789}$& 96 & 58 & 9.96 & 6.12 & 38.59\%\\
		\cmidrule[1.1pt]{1-7}
		&&&&&&\\[-1.05em]
		$\langle 3,3,3;23 \rangle$~\cite{benson2015framework} & $n^{\log_{3} 23} \approx n^{2.85404}$ & 87 & 66 & 7.21 & 5.71 & 20.79\%\\
		\hline
		&&&&&&\\[-1.05em]
		$\langle 3,3,3;23 \rangle$~\cite{benson2015framework} & $n^{\log_{3} 23} \approx n^{2.85404}$ & 88 & 65 & 7.29 & 5.64 & 22.55\%\\
		\hline
		&&&&&&\\[-1.05em]
		$\langle 3,3,3;23 \rangle$~\cite{benson2015framework} & $n^{\log_{3} 23} \approx n^{2.85404}$ & 89 & 65 & 7.36 & 5.64 & 23.3\%\\
		\hline
		&&&&&&\\[-1.05em]
		$\langle 3,3,3;23 \rangle$~\cite{benson2015framework} &$n^{\log_{3} 23} \approx n^{2.85404}$& 97 & 61 & 7.93 & 5.36 & 32.43\%\\
		\hline
		&&&&&&\\[-1.05em]
		$\langle 3,3,3;23 \rangle$~\cite{benson2015framework} &$n^{\log_{3} 23} \approx n^{2.85404}$& 166 & 73 & 12.86 & 6.21 & 51.67\%\\
		\hline
		&&&&&&\\[-1.05em]
		$\langle 3,3,3;23 \rangle$~\cite{laderman1976noncommutative} &$n^{\log_{3} 23} \approx n^{2.85404}$& 98 & 74 & 8 & 6.29 & 21.43\%\\
		\hline
		&&&&&&\\[-1.05em]
		$\langle 3,3,3;23 \rangle$~\cite{smirnov2013bilinear} &$n^{\log_{3} 23} \approx n^{2.85404}$& 84 & 68 & 7 & 5.86 & 16.33\%\\
		\cmidrule[1.1pt]{1-7}
		&&&&&&\\[-1.05em]
		$\langle 4,4,2;26 \rangle$~\cite{benson2015framework} &$n^{\log_{32} 26^{3}} \approx n^{2.82026}$& 235 & 105 ($\star$) & 18.1 & 7.81 & 56.84\%\\
		\cmidrule[1.1pt]{1-7}
		&&&&&&\\[-1.05em]
		$\langle 4,3,3;29 \rangle$~\cite{benson2015framework} &$n^{\log_{36} 29^{3}} \approx n^{2.81898}$& 164 & 102 & 10.27 & 6.73 & 34.49\%\\
		\hline
		&&&&&&\\[-1.05em]
		$\langle 3,4,3;29 \rangle$~\cite{smirnov2017several} &$n^{\log_{36} 29^{3}} \approx n^{2.81898}$& 137 & 109 & 8.54 & 6.96 & 18.46\%\\
		\hline
		&&&&&&\\[-1.05em]
		$\langle 3,4,3;29 \rangle$~\cite{benson2015framework} &$n^{\log_{36} 29^{3}} \approx n^{2.81898}$& 167 & 105 & 10.27 & 6.73 & 34.49\%\\
		\cmidrule[1.1pt]{1-7}
		&&&&&&\\[-1.05em]
		$\langle 3,5,3;36 \rangle$~\cite{smirnov2017several} &$n^{\log_{45} 36^{3}} \approx n^{2.82414}$& 199 & 139 & 9.62 & 6.87 & 28.6\%\\
		\cmidrule[1.1pt]{1-7}
		&&&&&&\\[-1.05em]
		$\langle 6,3,3;40 \rangle$~\cite{smirnov2013bilinear} &$n^{\log_{54} 40^{3}} \approx n^{2.77429}$& 1246 & 190 ($\star$) & 55.63 & 8.9 & 84.01\%\\
		\hline
		&&&&&&\\[-1.05em]
		$\langle 3,3,6;40 \rangle$~\cite{tichavsky2017numerical} &$n^{\log_{54} 40^{3}} \approx n^{2.77429}$& 1822 & 190 ($\star$) & 79.28 & 8.9 & 88.78\%\\
		\hline
	\end{tabular}
	\caption*{\footnotesize \textmd{($\star$) Denotes algorithms with non-singular values, where the result of Algorithm~\ref{BSparseAlg} was better than those of the exhaustive algorithms.}}
	}
\end{table*}

\section{Application and resulting algorithms}\label{sec:applications-results}

Table~\ref{tab:decomposedalgs} contains a list of alternative basis algorithms found using our new methods. All of the algorithms used were taken from the repository of Ballard and Benson~\cite{benson2015framework}\footnote{The algorithms can be found at github.com/arbenson/fast-matmul}. The alternative basis algorithms obtained represent a significant improvement over the original versions, with the reduction in the leading coefficient ranging between 15\% and 88\%. Almost all of the results were found using our exhaustive methods (Algorithms~\ref{SIV1Alg} and \ref{SIV2Alg}). In certain cases (marked \textmd{($\star$)}), where the $U,\,V,\,W$ matrices contain non-singular values, our search heuristic's (Algorithm~\ref{BSparseAlg}) result exceeded those of our exhaustive algorithms. For example, bases obtained for the $\strass{4}{4}{2}{26}$-algorithm by Algorithms~\ref{SIV1Alg} and \ref{SIV2Alg} reduced the number of arithmetic operations from 235 to 110, while Algorithm~\ref{BSparseAlg} reduced the number of arithmetic operations even further, to 105.

\paragraph{Comparison of different search methods.}
The exhaustive algorithms (Algorithms~\ref{SIV1Alg}, \ref{SIV2Alg}) solve the SIV problem. Their proof of correctness, coupled with that of Gottlieb and Neylon's algorithm, guarantee that they obtain decompositions minimizing the number of non-zero entries. As MS and SIV are both NP-Hard problems, these algorithms exhibit an exponential worst-case complexity. For this reason, the decomposition of some of the larger instances required the use of Mira supercomputer.
However after some tuning of Algorithms 3 and 4 (see Section~\ref{subsec:nathan-implementation}) and the implementation of Algorithm~\ref{BSparseAlg} using Z3, all decompositions completed on a PC within a reasonable time. Specifically, all runs of Algorithms 3 and 4 completed within 40 minutes, while Algorithm 6 took less than one minute, on a PC\footnote{Matebook X (i7-7500U CPU and 8GB RAM)}.
It should be remembered that Algorithms 3 and 4 guarantee optimal sparsification, while Algorithm 6 has no such guarantee. However, in all cases, Algorithm 6 ran much faster and produced an equally good decomposition, with better results when there were non-singular values.

\section{Discussion and Future Work}\label{sec:discussion}
We have improved the leading coefficient of several fast matrix multiplication algorithms by introducing new methods to solve to sparsify the encoding/decoding matrices of fast matrix multiplication algorithms. 
The number of arithmetic operations depends on both both non-zero and non-singular entries. This means that in order to minimize the arithmetic complexity, the sum of both non-zero \textit{and} non-singular entries should be minimized, otherwise an optimal sparsification may result in a 2-approximation of the minimal number of arithmetic operations when matrix entries are not limited to $0,\,\pm 1$. Further work is required in order to find a provably optimal algorithm which minimizes both non-zero and non-singleton values.

We attempted sparsification of additional algorithms for larger dimensions (e.g., Pan's $\squarestrass{44}{36133}$-algorithm~\cite{pan1982trilinear}, which is asymptotically faster than those presented here). However, the
size of the base case of these algorithms led to prohibitively long
runtimes.

\balance

The methods presented in this paper apply to finding square invertible matrices solving the MS problem. Other classes of sparse decompositions exist which do not fall within this category. For example, Beniamini and Schwartz's~\cite{beniamini2019faster} decomposed recursive-bilinear framework relies upon decompositions in which the sparsifying matrix may be rectangular, rather than square. Some of the leading coefficients in~\cite{beniamini2019faster} are better than those presented here. For example, they obtained a leading coeffcient of 2 for a $\strass{3}{3}{3}{23}$-algorithm of~\cite{benson2015framework} a $\strass{4}{3}{3}{29}$-algorithm of~\cite{smirnov2017several}, compared to our values $5.36$ and $6.96$ respectively. However, the arithmetic overhead of basis transformation in Karstadt and Schwartz \cite{karstadt2017matrix, karstadt2020matrix} (and therefore here as well) is $O\left(n^{2} \log n\right)$, whereas
in~\cite{beniamini2019faster} it may be larger. Note also that the decomposition heuristic of~\cite{beniamini2019faster} does not always guarantee optimality. Further work is required to find new decomposition methods for such settings.

\begin{anonsuppress}
\section{Acknowledgements}
We thank Austin R. Benson for providing details regarding the $\strass{2}{3}{2}{11}$-algorithm. This research used resources of the Argonne Leadership Computing Facility, which is a DOE Office of Science User Facility supported under Contract DE-AC02-06CH11357. This work was supported by the PetaCloud industry-academia consortium. This research was supported by a grant from the United States-Israel Bi-national Science Foundation, Jerusalem, Israel. This work was supported by the HUJI Cyber Security Research Center in conjunction with the Israel National Cyber Bureau in the Prime Minister's Office. This project has received funding from the European Research Council (ERC) under the European Union's Horizon 2020 research and innovation programme (grant agreement No 818252).
\end{anonsuppress}

\bibliographystyle{ACM-Reference-Format}

\bibliography{main}
\newpage
\clearpage
\appendix
\section{Samples of Alternative Basis Algorithms}\label{sec:results-appendix}
In this section we present the encoding/decoding matrices of the alternative basis algorithms listed in Table~\ref{tab:decomposedalgs}. To verify the correctness of these algorithms, recall Corollary~\ref{cor:alt-to-mul} and use the following fact:

\begin{fact} (Triple product condition).
\cite{brent1970algorithms, knuth1981art}\label{fact:Triple-product-condition} Let $R$ be a ring, and let $U\in R^{t\times n\cdot m},\,V\in R^{t\times m\cdot k},\,W\in R^{t\times n\cdot k}$. Then $\encdec{U}{V}{W} $ are encoding/decoding matrices of an $\strass{n}{m}{k}{t}$-algorithm if and only if:
\begin{align*}
\forall i_{1},k_{1}\in\left[n\right],\,j_{1},i_{2}\in\left[m\right],\,j_{2},k_{2}\in\left[k\right]\: \\
\sum_{r=1}^{t}U_{r,\left(i_{1},i_{2}\right)}V_{r,\left(j_{1},j_{2}\right)}W_{r,\left(k_{1},k_{2}\right)}=\delta_{i_{1},k_{1}}\delta_{i_{2},j_{1}}\delta_{j_{2},k_{2}}
\end{align*}
where $\delta_{i,j}=1$ if $i=j$ and $0$ otherwise.

\end{fact}

\newcommand\VRule[1][\arrayrulewidth]{\vrule width #1}

\subsection{A sample of Algorithms}
\begin{center}
\captionof{table}{$\strass{3}{2}{3}{15}$-algorithm~\cite{benson2015framework}}
\resizebox{\linewidth}{!}{
\begin{tabular}{!{\VRule[2pt]}*{6}c!{\VRule[2pt]}*{6}c!{\VRule[2pt]} *{9}c!{\VRule[2pt]}}
\multicolumn{6}{c}{$U_{\phi}$} & \multicolumn{6}{c}{$V_{\psi}$} & \multicolumn{9}{c}{$W_{\upsilon}$}\\
\specialrule{1pt}{0pt}{0pt}
$0$ & $-1$ & $0$ & $1$ & $0$ & $1$ & $1$ & $0$ & $0$ & $0$ & $0$ & $-1$ & $0$ & $0$ & $1$ & $0$ & $0$ & $0$ & $0$ & $0$ & $0$ \\
$0$ & $0$ & $0$ & $1$ & $0$ & $0$ & $0$ & $0$ & $0$ & $0$ & $0$ & $1$ & $0$ & $0$ & $0$ & $0$ & $0$ & $0$ & $0$ & $1$ & $0$ \\
$0$ & $0$ & $1$ & $1$ & $0$ & $0$ & $0$ & $0$ & $-1$ & $0$ & $1$ & $0$ & $-1$ & $0$ & $0$ & $0$ & $0$ & $0$ & $-1$ & $0$ & $0$ \\
$0$ & $0$ & $0$ & $0$ & $-1$ & $0$ & $0$ & $0$ & $0$ & $1$ & $0$ & $0$ & $0$ & $0$ & $0$ & $0$ & $0$ & $0$ & $0$ & $0$ & $-1$ \\
$0$ & $0$ & $0$ & $0$ & $0$ & $-1$ & $0$ & $0$ & $0$ & $0$ & $-1$ & $0$ & $0$ & $0$ & $0$ & $0$ & $0$ & $0$ & $-1$ & $0$ & $0$ \\
$1$ & $0$ & $0$ & $0$ & $-1$ & $0$ & $0$ & $0$ & $-1$ & $0$ & $0$ & $0$ & $0$ & $0$ & $0$ & $0$ & $0$ & $-1$ & $0$ & $0$ & $0$ \\
$0$ & $-1$ & $0$ & $0$ & $1$ & $0$ & $0$ & $-1$ & $0$ & $-1$ & $0$ & $1$ & $0$ & $0$ & $0$ & $1$ & $0$ & $0$ & $0$ & $0$ & $0$ \\
$0$ & $-1$ & $0$ & $0$ & $0$ & $0$ & $1$ & $0$ & $0$ & $-1$ & $0$ & $0$ & $0$ & $0$ & $0$ & $0$ & $0$ & $-1$ & $-1$ & $0$ & $-1$ \\
$1$ & $0$ & $0$ & $1$ & $0$ & $0$ & $0$ & $0$ & $1$ & $0$ & $-1$ & $1$ & $0$ & $0$ & $1$ & $0$ & $0$ & $0$ & $0$ & $-1$ & $0$ \\
$-1$ & $0$ & $0$ & $0$ & $1$ & $-1$ & $0$ & $-1$ & $0$ & $0$ & $1$ & $0$ & $0$ & $-1$ & $0$ & $0$ & $0$ & $0$ & $-1$ & $0$ & $0$ \\
$0$ & $0$ & $-1$ & $0$ & $0$ & $0$ & $1$ & $0$ & $0$ & $0$ & $0$ & $0$ & $0$ & $0$ & $0$ & $0$ & $-1$ & $0$ & $0$ & $0$ & $0$ \\
$1$ & $0$ & $0$ & $0$ & $0$ & $0$ & $0$ & $-1$ & $0$ & $0$ & $0$ & $1$ & $0$ & $0$ & $0$ & $0$ & $-1$ & $0$ & $0$ & $1$ & $-1$ \\
$0$ & $0$ & $-1$ & $-1$ & $0$ & $-1$ & $1$ & $0$ & $0$ & $-1$ & $-1$ & $0$ & $0$ & $-1$ & $0$ & $0$ & $0$ & $0$ & $0$ & $0$ & $0$ \\
$0$ & $0$ & $1$ & $0$ & $-1$ & $0$ & $0$ & $0$ & $-1$ & $-1$ & $0$ & $0$ & $0$ & $0$ & $0$ & $1$ & $0$ & $0$ & $0$ & $0$ & $-1$ \\
$0$ & $1$ & $0$ & $0$ & $0$ & $-1$ & $0$ & $-1$ & $0$ & $0$ & $0$ & $0$ & $-1$ & $0$ & $0$ & $0$ & $0$ & $0$ & $0$ & $1$ & $0$ \\
\specialrule{1pt}{0pt}{0pt}
\multicolumn{6}{c}{$\phi$} & \multicolumn{6}{c}{$\psi$} & \multicolumn{9}{c}{$\upsilon^{-T}$}\\
\specialrule{1pt}{0pt}{0pt}
$0$ & $0$ & $0$ & $0$ & $1$ & $-1$ & $0$ & $0$ & $1$ & $0$ & $0$ & $1$ & $0$ & $0$ & $0$ & $0$ & $1$ & $0$ & $0$ & $-1$ & $0$ \\
$0$ & $0$ & $0$ & $-1$ & $0$ & $0$ & $0$ & $0$ & $0$ & $0$ & $-1$ & $0$ & $-1$ & $0$ & $0$ & $0$ & $0$ & $0$ & $0$ & $0$ & $0$ \\
$0$ & $0$ & $0$ & $0$ & $0$ & $-1$ & $-1$ & $0$ & $0$ & $0$ & $0$ & $0$ & $0$ & $0$ & $-1$ & $0$ & $0$ & $0$ & $0$ & $0$ & $0$ \\
$-1$ & $0$ & $1$ & $0$ & $0$ & $1$ & $0$ & $0$ & $1$ & $-1$ & $0$ & $1$ & $0$ & $0$ & $-1$ & $0$ & $0$ & $-1$ & $0$ & $0$ & $0$ \\
$0$ & $0$ & $-1$ & $0$ & $1$ & $-1$ & $-1$ & $-1$ & $0$ & $0$ & $-1$ & $0$ & $0$ & $0$ & $-1$ & $0$ & $0$ & $0$ & $-1$ & $0$ & $-1$ \\
$1$ & $-1$ & $-1$ & $0$ & $0$ & $0$ & $0$ & $0$ & $1$ & $0$ & $-1$ & $0$ & $-1$ & $1$ & $0$ & $-1$ & $1$ & $0$ & $-1$ & $0$ & $0$ \\
  &   &   &   &   &   &   &   &   &   &   &   & $1$ & $-1$ & $0$ & $0$ & $-1$ & $0$ & $0$ & $0$ & $0$ \\
  &   &   &   &   &   &   &   &   &   &   &   & $0$ & $0$ & $-1$ & $0$ & $0$ & $0$ & $0$ & $-1$ & $0$ \\
  &   &   &   &   &   &   &   &   &   &   &   & $0$ & $0$ & $-1$ & $0$ & $0$ & $-1$ & $1$ & $0$ & $0$ \\
\specialrule{1pt}{0pt}{0pt}
\end{tabular}
}
\end{center}

\begin{center}
\captionof{table}{$\strass{4}{2}{3}{20}$-algorithm~\cite{smirnov2013bilinear}}
\resizebox{\linewidth}{!}{
\begin{tabular}{!{\VRule[2pt]}*{8}c!{\VRule[2pt]}*{6}c!{\VRule[2pt]} *{12}c!{\VRule[2pt]}}
\multicolumn{8}{c}{$U_{\phi}$} & \multicolumn{6}{c}{$V_{\psi}$} & \multicolumn{12}{c}{$W_{\upsilon}$}\\
\specialrule{1pt}{0pt}{0pt}
$0$ & $0$ & $-1$ & $0$ & $0$ & $1$ & $0$ & $0$ & $-1$ & $0$ & $0$ & $0$ & $0$ & $0$ & $0$ & $0$ & $0$ & $0$ & $0$ & $0$ & $0$ & $1$ & $0$ & $0$ & $-1$ & $0$ \\
$0$ & $0$ & $-1$ & $0$ & $0$ & $0$ & $0$ & $1$ & $0$ & $1$ & $0$ & $0$ & $1$ & $0$ & $0$ & $0$ & $0$ & $0$ & $0$ & $-1$ & $0$ & $0$ & $0$ & $0$ & $1$ & $0$ \\
$0$ & $0$ & $0$ & $0$ & $0$ & $0$ & $-1$ & $1$ & $0$ & $0$ & $0$ & $1$ & $-1$ & $0$ & $0$ & $0$ & $0$ & $0$ & $0$ & $0$ & $0$ & $0$ & $0$ & $1$ & $0$ & $0$ \\
$0$ & $0$ & $0$ & $0$ & $1$ & $0$ & $0$ & $0$ & $0$ & $1$ & $0$ & $0$ & $0$ & $1$ & $0$ & $0$ & $0$ & $0$ & $1$ & $0$ & $0$ & $0$ & $1$ & $0$ & $0$ & $-1$ \\
$0$ & $0$ & $0$ & $0$ & $0$ & $-1$ & $0$ & $0$ & $0$ & $0$ & $0$ & $0$ & $0$ & $1$ & $0$ & $0$ & $0$ & $0$ & $0$ & $0$ & $0$ & $0$ & $1$ & $0$ & $0$ & $0$ \\
$0$ & $-1$ & $0$ & $0$ & $0$ & $0$ & $0$ & $-1$ & $0$ & $-1$ & $0$ & $0$ & $0$ & $0$ & $-1$ & $0$ & $0$ & $0$ & $0$ & $0$ & $0$ & $0$ & $0$ & $0$ & $0$ & $0$ \\
$0$ & $0$ & $0$ & $-1$ & $0$ & $0$ & $0$ & $0$ & $1$ & $0$ & $0$ & $0$ & $0$ & $1$ & $0$ & $0$ & $0$ & $0$ & $-1$ & $0$ & $0$ & $0$ & $0$ & $0$ & $0$ & $0$ \\
$0$ & $0$ & $0$ & $0$ & $1$ & $0$ & $0$ & $-1$ & $0$ & $-1$ & $0$ & $1$ & $-1$ & $0$ & $0$ & $0$ & $0$ & $-1$ & $0$ & $0$ & $0$ & $0$ & $0$ & $1$ & $-1$ & $0$ \\
$0$ & $0$ & $0$ & $-1$ & $0$ & $-1$ & $0$ & $0$ & $0$ & $0$ & $1$ & $-1$ & $0$ & $1$ & $0$ & $1$ & $0$ & $0$ & $0$ & $0$ & $0$ & $0$ & $0$ & $0$ & $0$ & $0$ \\
$0$ & $-1$ & $0$ & $0$ & $0$ & $0$ & $-1$ & $0$ & $-1$ & $0$ & $0$ & $0$ & $-1$ & $0$ & $0$ & $0$ & $0$ & $0$ & $0$ & $0$ & $1$ & $0$ & $0$ & $-1$ & $0$ & $1$ \\
$-1$ & $0$ & $0$ & $0$ & $0$ & $-1$ & $0$ & $0$ & $0$ & $0$ & $-1$ & $0$ & $0$ & $0$ & $0$ & $0$ & $0$ & $0$ & $0$ & $-1$ & $0$ & $0$ & $0$ & $0$ & $0$ & $0$ \\
$0$ & $0$ & $0$ & $0$ & $0$ & $0$ & $-1$ & $0$ & $0$ & $0$ & $0$ & $0$ & $-1$ & $1$ & $0$ & $0$ & $0$ & $0$ & $0$ & $0$ & $0$ & $0$ & $0$ & $0$ & $0$ & $-1$ \\
$0$ & $0$ & $0$ & $0$ & $0$ & $1$ & $0$ & $-1$ & $0$ & $0$ & $0$ & $1$ & $0$ & $0$ & $0$ & $0$ & $0$ & $0$ & $0$ & $0$ & $0$ & $0$ & $0$ & $0$ & $1$ & $0$ \\
$-1$ & $0$ & $0$ & $0$ & $0$ & $0$ & $0$ & $-1$ & $-1$ & $0$ & $0$ & $-1$ & $0$ & $0$ & $0$ & $0$ & $0$ & $-1$ & $0$ & $0$ & $0$ & $0$ & $0$ & $0$ & $0$ & $0$ \\
$0$ & $0$ & $0$ & $0$ & $1$ & $-1$ & $0$ & $0$ & $0$ & $0$ & $-1$ & $1$ & $0$ & $0$ & $0$ & $0$ & $0$ & $0$ & $0$ & $0$ & $0$ & $1$ & $-1$ & $0$ & $0$ & $0$ \\
$0$ & $0$ & $0$ & $-1$ & $0$ & $0$ & $-1$ & $0$ & $0$ & $0$ & $-1$ & $0$ & $0$ & $0$ & $0$ & $0$ & $1$ & $0$ & $0$ & $0$ & $0$ & $0$ & $0$ & $0$ & $0$ & $0$ \\
$1$ & $0$ & $0$ & $0$ & $0$ & $0$ & $1$ & $0$ & $0$ & $0$ & $-1$ & $1$ & $0$ & $-1$ & $-1$ & $0$ & $0$ & $0$ & $0$ & $0$ & $0$ & $0$ & $0$ & $1$ & $0$ & $0$ \\
$0$ & $0$ & $1$ & $0$ & $0$ & $0$ & $0$ & $0$ & $0$ & $-1$ & $0$ & $0$ & $0$ & $0$ & $0$ & $1$ & $0$ & $0$ & $0$ & $0$ & $0$ & $0$ & $1$ & $0$ & $0$ & $0$ \\
$0$ & $0$ & $0$ & $0$ & $-1$ & $0$ & $1$ & $0$ & $0$ & $0$ & $1$ & $0$ & $-1$ & $1$ & $0$ & $0$ & $0$ & $0$ & $0$ & $0$ & $1$ & $0$ & $0$ & $0$ & $0$ & $0$ \\
$0$ & $-1$ & $0$ & $0$ & $0$ & $0$ & $0$ & $0$ & $0$ & $-1$ & $0$ & $0$ & $-1$ & $0$ & $0$ & $0$ & $1$ & $0$ & $0$ & $0$ & $0$ & $0$ & $0$ & $0$ & $0$ & $-1$ \\
\specialrule{1pt}{0pt}{0pt}
\multicolumn{8}{c}{$\phi$} & \multicolumn{6}{c}{$\psi$} & \multicolumn{12}{c}{$\upsilon^{-T}$}\\
\specialrule{1pt}{0pt}{0pt}
$-1$ & $0$ & $0$ & $0$ & $0$ & $0$ & $0$ & $-1$ & $-1$ & $0$ & $0$ & $0$ & $0$ & $0$ & $0$ & $0$ & $0$ & $-1$ & $0$ & $0$ & $0$ & $0$ & $0$ & $0$ & $0$ & $0$ \\
$-1$ & $0$ & $-1$ & $0$ & $0$ & $0$ & $0$ & $0$ & $-1$ & $-1$ & $0$ & $0$ & $0$ & $0$ & $1$ & $0$ & $0$ & $0$ & $0$ & $0$ & $0$ & $0$ & $0$ & $0$ & $0$ & $0$ \\
$1$ & $0$ & $0$ & $0$ & $1$ & $0$ & $0$ & $0$ & $0$ & $0$ & $0$ & $0$ & $0$ & $-1$ & $0$ & $0$ & $1$ & $0$ & $0$ & $0$ & $0$ & $0$ & $0$ & $0$ & $0$ & $0$ \\
$-1$ & $-1$ & $0$ & $0$ & $0$ & $0$ & $0$ & $0$ & $1$ & $0$ & $0$ & $-1$ & $0$ & $-1$ & $0$ & $0$ & $0$ & $0$ & $0$ & $0$ & $0$ & $0$ & $0$ & $-1$ & $1$ & $0$ \\
$1$ & $0$ & $0$ & $0$ & $0$ & $0$ & $0$ & $0$ & $1$ & $1$ & $-1$ & $0$ & $0$ & $0$ & $1$ & $-1$ & $0$ & $0$ & $0$ & $0$ & $0$ & $0$ & $0$ & $0$ & $0$ & $0$ \\
$1$ & $0$ & $0$ & $0$ & $0$ & $-1$ & $0$ & $0$ & $1$ & $0$ & $0$ & $0$ & $1$ & $0$ & $0$ & $0$ & $0$ & $0$ & $0$ & $0$ & $0$ & $0$ & $0$ & $1$ & $-1$ & $-1$ \\
$1$ & $0$ & $0$ & $-1$ & $0$ & $0$ & $0$ & $0$ &   &   &   &   &   &   & $0$ & $0$ & $-1$ & $-1$ & $0$ & $1$ & $0$ & $0$ & $0$ & $0$ & $1$ & $0$ \\
$1$ & $0$ & $0$ & $0$ & $0$ & $0$ & $-1$ & $0$ &   &   &   &   &   &   & $0$ & $0$ & $0$ & $0$ & $0$ & $0$ & $-1$ & $1$ & $1$ & $1$ & $-1$ & $-1$ \\
  &   &   &   &   &   &   &   &   &   &   &   &   &   & $-1$ & $0$ & $0$ & $0$ & $0$ & $0$ & $0$ & $1$ & $0$ & $0$ & $0$ & $0$ \\
  &   &   &   &   &   &   &   &   &   &   &   &   &   & $0$ & $0$ & $0$ & $-1$ & $0$ & $0$ & $0$ & $0$ & $0$ & $0$ & $1$ & $0$ \\
  &   &   &   &   &   &   &   &   &   &   &   &   &   & $0$ & $0$ & $0$ & $0$ & $0$ & $0$ & $0$ & $0$ & $1$ & $1$ & $-1$ & $-1$ \\
  &   &   &   &   &   &   &   &   &   &   &   &   &   & $0$ & $0$ & $1$ & $1$ & $-1$ & $-1$ & $0$ & $0$ & $0$ & $0$ & $0$ & $0$ \\
\specialrule{1pt}{0pt}{0pt}
\end{tabular}
}
\end{center}
\begin{center}
\captionof{table}{$\strass{3}{3}{3}{23}$-algorithm~\cite{smirnov2013bilinear}}
\resizebox{\linewidth}{!}{
\begin{tabular}{!{\VRule[2pt]}*{9}c!{\VRule[2pt]}*{9}c!{\VRule[2pt]} *{9}c!{\VRule[2pt]}}
\multicolumn{9}{c}{$U_{\phi}$} & \multicolumn{9}{c}{$V_{\psi}$} & \multicolumn{9}{c}{$W_{\upsilon}$}\\
\specialrule{1pt}{0pt}{0pt}
$0$ & $0$ & $0$ & $0$ & $0$ & $0$ & $0$ & $0$ & $1$ & $-1$ & $0$ & $0$ & $0$ & $0$ & $0$ & $0$ & $-1$ & $0$ & $0$ & $-1$ & $0$ & $0$ & $-1$ & $0$ & $0$ & $1$ & $0$ \\
$0$ & $0$ & $0$ & $0$ & $0$ & $0$ & $0$ & $1$ & $1$ & $0$ & $0$ & $0$ & $0$ & $0$ & $0$ & $-1$ & $0$ & $0$ & $0$ & $0$ & $0$ & $1$ & $0$ & $0$ & $0$ & $-1$ & $0$ \\
$0$ & $1$ & $0$ & $0$ & $-1$ & $0$ & $0$ & $0$ & $0$ & $0$ & $1$ & $0$ & $0$ & $0$ & $0$ & $0$ & $0$ & $1$ & $0$ & $0$ & $0$ & $0$ & $0$ & $1$ & $-1$ & $0$ & $0$ \\
$1$ & $0$ & $0$ & $0$ & $0$ & $0$ & $-1$ & $0$ & $0$ & $0$ & $0$ & $-1$ & $0$ & $0$ & $0$ & $0$ & $0$ & $0$ & $1$ & $0$ & $0$ & $0$ & $0$ & $0$ & $0$ & $0$ & $0$ \\
$0$ & $0$ & $0$ & $0$ & $1$ & $0$ & $1$ & $0$ & $0$ & $0$ & $1$ & $0$ & $0$ & $0$ & $0$ & $0$ & $0$ & $0$ & $0$ & $0$ & $0$ & $1$ & $0$ & $0$ & $-1$ & $0$ & $0$ \\
$0$ & $0$ & $0$ & $-1$ & $0$ & $1$ & $0$ & $0$ & $0$ & $0$ & $0$ & $0$ & $0$ & $0$ & $-1$ & $-1$ & $0$ & $0$ & $0$ & $0$ & $-1$ & $0$ & $0$ & $0$ & $0$ & $1$ & $1$ \\
$-1$ & $0$ & $0$ & $0$ & $0$ & $0$ & $0$ & $0$ & $-1$ & $0$ & $0$ & $0$ & $0$ & $0$ & $0$ & $0$ & $1$ & $0$ & $0$ & $0$ & $0$ & $0$ & $0$ & $0$ & $0$ & $-1$ & $0$ \\
$0$ & $0$ & $-1$ & $0$ & $-1$ & $0$ & $-1$ & $0$ & $1$ & $0$ & $0$ & $0$ & $0$ & $0$ & $0$ & $0$ & $1$ & $0$ & $0$ & $-1$ & $0$ & $0$ & $0$ & $0$ & $0$ & $0$ & $0$ \\
$0$ & $0$ & $0$ & $0$ & $0$ & $0$ & $0$ & $1$ & $0$ & $0$ & $0$ & $0$ & $0$ & $1$ & $0$ & $0$ & $0$ & $-1$ & $0$ & $0$ & $-1$ & $1$ & $0$ & $0$ & $0$ & $0$ & $0$ \\
$0$ & $0$ & $-1$ & $0$ & $0$ & $0$ & $0$ & $1$ & $0$ & $0$ & $0$ & $0$ & $0$ & $-1$ & $0$ & $1$ & $0$ & $1$ & $0$ & $0$ & $0$ & $1$ & $0$ & $0$ & $0$ & $0$ & $0$ \\
$0$ & $0$ & $-1$ & $0$ & $0$ & $0$ & $0$ & $0$ & $-1$ & $0$ & $0$ & $0$ & $0$ & $0$ & $-1$ & $0$ & $0$ & $0$ & $0$ & $1$ & $0$ & $-1$ & $0$ & $-1$ & $0$ & $0$ & $0$ \\
$0$ & $0$ & $0$ & $0$ & $0$ & $1$ & $0$ & $0$ & $0$ & $0$ & $0$ & $0$ & $0$ & $1$ & $0$ & $0$ & $0$ & $0$ & $0$ & $0$ & $0$ & $0$ & $0$ & $0$ & $0$ & $0$ & $1$ \\
$0$ & $0$ & $0$ & $-1$ & $0$ & $0$ & $0$ & $0$ & $1$ & $-1$ & $0$ & $0$ & $1$ & $0$ & $-1$ & $0$ & $0$ & $0$ & $0$ & $0$ & $0$ & $0$ & $1$ & $1$ & $0$ & $0$ & $0$ \\
$0$ & $0$ & $-1$ & $0$ & $-1$ & $0$ & $0$ & $0$ & $0$ & $0$ & $0$ & $-1$ & $0$ & $0$ & $0$ & $0$ & $1$ & $-1$ & $0$ & $0$ & $0$ & $0$ & $0$ & $1$ & $0$ & $0$ & $0$ \\
$0$ & $0$ & $0$ & $0$ & $0$ & $0$ & $-1$ & $0$ & $1$ & $0$ & $0$ & $0$ & $1$ & $0$ & $0$ & $0$ & $1$ & $0$ & $0$ & $0$ & $0$ & $0$ & $-1$ & $0$ & $0$ & $0$ & $0$ \\
$0$ & $0$ & $0$ & $0$ & $0$ & $1$ & $0$ & $-1$ & $0$ & $0$ & $0$ & $0$ & $1$ & $-1$ & $0$ & $0$ & $0$ & $0$ & $1$ & $0$ & $0$ & $0$ & $0$ & $0$ & $0$ & $0$ & $0$ \\
$0$ & $-1$ & $0$ & $0$ & $0$ & $1$ & $0$ & $0$ & $0$ & $0$ & $0$ & $0$ & $0$ & $-1$ & $1$ & $1$ & $0$ & $0$ & $-1$ & $0$ & $0$ & $0$ & $0$ & $0$ & $0$ & $0$ & $1$ \\
$0$ & $0$ & $0$ & $1$ & $0$ & $-1$ & $0$ & $1$ & $1$ & $1$ & $0$ & $0$ & $0$ & $0$ & $0$ & $-1$ & $0$ & $0$ & $0$ & $0$ & $0$ & $0$ & $0$ & $0$ & $0$ & $1$ & $0$ \\
$0$ & $0$ & $0$ & $0$ & $1$ & $0$ & $0$ & $0$ & $0$ & $0$ & $0$ & $0$ & $0$ & $0$ & $0$ & $0$ & $0$ & $-1$ & $0$ & $0$ & $0$ & $0$ & $0$ & $0$ & $1$ & $0$ & $0$ \\
$0$ & $0$ & $0$ & $1$ & $0$ & $0$ & $-1$ & $0$ & $0$ & $0$ & $0$ & $0$ & $1$ & $0$ & $0$ & $0$ & $0$ & $0$ & $0$ & $0$ & $0$ & $0$ & $1$ & $0$ & $1$ & $0$ & $-1$ \\
$0$ & $0$ & $0$ & $0$ & $0$ & $0$ & $1$ & $0$ & $0$ & $0$ & $0$ & $1$ & $1$ & $0$ & $0$ & $0$ & $0$ & $0$ & $0$ & $0$ & $0$ & $0$ & $0$ & $0$ & $1$ & $0$ & $-1$ \\
$-1$ & $0$ & $0$ & $0$ & $0$ & $0$ & $0$ & $-1$ & $0$ & $0$ & $-1$ & $0$ & $0$ & $0$ & $0$ & $0$ & $0$ & $0$ & $0$ & $0$ & $1$ & $0$ & $0$ & $0$ & $0$ & $0$ & $0$ \\
$0$ & $-1$ & $0$ & $0$ & $0$ & $0$ & $0$ & $0$ & $1$ & $0$ & $0$ & $0$ & $0$ & $0$ & $-1$ & $0$ & $0$ & $0$ & $0$ & $0$ & $0$ & $0$ & $-1$ & $-1$ & $0$ & $0$ & $0$ \\
\specialrule{1pt}{0pt}{0pt}
\multicolumn{9}{c}{$\phi$} & \multicolumn{9}{c}{$\psi$} & \multicolumn{9}{c}{$\upsilon^{-T}$}\\
\specialrule{1pt}{0pt}{0pt}
$0$ & $0$ & $0$ & $1$ & $0$ & $0$ & $0$ & $0$ & $0$ & $0$ & $0$ & $0$ & $1$ & $0$ & $0$ & $-1$ & $0$ & $0$ & $0$ & $0$ & $0$ & $0$ & $0$ & $1$ & $0$ & $0$ & $0$ \\
$0$ & $1$ & $0$ & $0$ & $0$ & $0$ & $0$ & $0$ & $0$ & $0$ & $-1$ & $0$ & $0$ & $0$ & $0$ & $0$ & $0$ & $0$ & $0$ & $0$ & $0$ & $0$ & $0$ & $0$ & $-1$ & $0$ & $1$ \\
$0$ & $0$ & $0$ & $0$ & $0$ & $0$ & $0$ & $1$ & $0$ & $0$ & $-1$ & $-1$ & $0$ & $0$ & $0$ & $0$ & $0$ & $0$ & $0$ & $0$ & $0$ & $0$ & $-1$ & $1$ & $0$ & $0$ & $0$ \\
$0$ & $0$ & $-1$ & $0$ & $0$ & $0$ & $0$ & $0$ & $0$ & $0$ & $0$ & $0$ & $0$ & $0$ & $0$ & $0$ & $0$ & $1$ & $0$ & $0$ & $0$ & $0$ & $0$ & $0$ & $0$ & $1$ & $0$ \\
$-1$ & $0$ & $0$ & $0$ & $0$ & $0$ & $-1$ & $-1$ & $0$ & $0$ & $0$ & $0$ & $0$ & $1$ & $1$ & $0$ & $-1$ & $0$ & $1$ & $0$ & $0$ & $0$ & $0$ & $0$ & $0$ & $0$ & $-1$ \\
$0$ & $0$ & $-1$ & $0$ & $1$ & $1$ & $0$ & $0$ & $0$ & $0$ & $0$ & $0$ & $-1$ & $0$ & $0$ & $0$ & $0$ & $0$ & $0$ & $0$ & $0$ & $0$ & $0$ & $0$ & $0$ & $0$ & $1$ \\
$1$ & $0$ & $0$ & $0$ & $0$ & $0$ & $0$ & $0$ & $0$ & $0$ & $0$ & $0$ & $1$ & $0$ & $0$ & $0$ & $-1$ & $0$ & $0$ & $1$ & $-1$ & $0$ & $0$ & $0$ & $0$ & $0$ & $1$ \\
$0$ & $0$ & $0$ & $0$ & $1$ & $0$ & $0$ & $0$ & $0$ & $1$ & $0$ & $0$ & $0$ & $0$ & $0$ & $0$ & $0$ & $0$ & $0$ & $0$ & $0$ & $1$ & $0$ & $0$ & $0$ & $0$ & $0$ \\
$0$ & $0$ & $0$ & $0$ & $0$ & $0$ & $0$ & $0$ & $1$ & $0$ & $1$ & $0$ & $0$ & $0$ & $1$ & $0$ & $0$ & $0$ & $0$ & $1$ & $0$ & $0$ & $0$ & $1$ & $0$ & $0$ & $0$ \\
\specialrule{1pt}{0pt}{0pt}
\end{tabular}
}
\end{center}
\begin{center}
\captionof{table}{$\strass{6}{3}{3}{40}$-algorithm~\cite{smirnov2013bilinear}}
\resizebox{\linewidth}{!}{
\begin{tabular}{!{\VRule[2pt]}*{18}c!{\VRule[2pt]}*{9}c!{\VRule[2pt]} *{18}c!{\VRule[2pt]}}
\multicolumn{18}{c}{$U_{\phi}$} & \multicolumn{9}{c}{$V_{\psi}$} & \multicolumn{18}{c}{$W_{\upsilon}$}\\
\specialrule{1pt}{0pt}{0pt}
$0$ & $0$ & $0$ & $0$ & $0$ & $0$ & $-1$ & $0$ & $0$ & $1$ & $0$ & $0$ & $0$ & $0$ & $0$ & $0$ & $0$ & $0$ & $0$ & $0$ & $0$ & $0$ & $0$ & $1$ & $-1$ & $0$ & $-1$ & $0$ & $0$ & $0$ & $0$ & $0$ & $0$ & $1$ & $0$ & $0$ & $0$ & $0$ & $0$ & $0$ & $0$ & $0$ & $0$ & $0$ & $0$ \\
$0$ & $0$ & $0$ & $0$ & $0$ & $0$ & $0$ & $0$ & $0$ & $-1$ & $0$ & $0$ & $0$ & $0$ & $0$ & $0$ & $0$ & $0$ & $0$ & $0$ & $0$ & $0$ & $-1$ & $0$ & $0$ & $0$ & $0$ & $0$ & $0$ & $0$ & $0$ & $0$ & $0$ & $0$ & $0$ & $0$ & $1$ & $0$ & $0$ & $0$ & $0$ & $0$ & $0$ & $0$ & $0$ \\
$0$ & $0$ & $0$ & $0$ & $0$ & $0$ & $-1$ & $0$ & $1$ & $0$ & $0$ & $0$ & $0$ & $0$ & $0$ & $0$ & $0$ & $0$ & $0$ & $0$ & $0$ & $1$ & $-1$ & $-1$ & $0$ & $-1$ & $0$ & $0$ & $0$ & $0$ & $0$ & $0$ & $0$ & $0$ & $0$ & $0$ & $0$ & $0$ & $0$ & $1$ & $0$ & $0$ & $0$ & $0$ & $0$ \\
$0$ & $0$ & $0$ & $0$ & $0$ & $0$ & $0$ & $0$ & $0$ & $0$ & $-1$ & $0$ & $0$ & $1$ & $0$ & $1$ & $0$ & $0$ & $0$ & $0$ & $0$ & $0$ & $0$ & $-1$ & $0$ & $-1$ & $0$ & $0$ & $0$ & $0$ & $0$ & $0$ & $0$ & $0$ & $0$ & $1$ & $0$ & $0$ & $0$ & $0$ & $0$ & $1$ & $0$ & $0$ & $0$ \\
$0$ & $0$ & $0$ & $0$ & $0$ & $0$ & $0$ & $0$ & $0$ & $0$ & $0$ & $0$ & $0$ & $0$ & $0$ & $1$ & $0$ & $0$ & $0$ & $0$ & $0$ & $-1$ & $0$ & $1$ & $0$ & $0$ & $0$ & $0$ & $0$ & $0$ & $0$ & $0$ & $0$ & $0$ & $0$ & $0$ & $0$ & $0$ & $0$ & $0$ & $0$ & $0$ & $0$ & $1$ & $0$ \\
$0$ & $0$ & $0$ & $0$ & $0$ & $0$ & $0$ & $0$ & $0$ & $0$ & $0$ & $-1$ & $0$ & $0$ & $0$ & $0$ & $0$ & $0$ & $0$ & $0$ & $0$ & $0$ & $-1$ & $0$ & $-1$ & $-1$ & $0$ & $0$ & $0$ & $0$ & $0$ & $0$ & $0$ & $0$ & $0$ & $0$ & $0$ & $-1$ & $0$ & $0$ & $0$ & $0$ & $0$ & $0$ & $-1$ \\
$0$ & $0$ & $0$ & $0$ & $0$ & $0$ & $0$ & $0$ & $0$ & $0$ & $0$ & $0$ & $1$ & $0$ & $0$ & $0$ & $0$ & $0$ & $0$ & $0$ & $0$ & $1$ & $-1$ & $0$ & $-1$ & $0$ & $0$ & $0$ & $0$ & $0$ & $0$ & $0$ & $0$ & $0$ & $0$ & $0$ & $0$ & $0$ & $0$ & $0$ & $0$ & $0$ & $-1$ & $0$ & $0$ \\
$0$ & $0$ & $0$ & $0$ & $0$ & $0$ & $0$ & $0$ & $0$ & $0$ & $0$ & $0$ & $-1$ & $0$ & $1$ & $0$ & $0$ & $0$ & $0$ & $0$ & $0$ & $0$ & $0$ & $0$ & $0$ & $-1$ & $-1$ & $0$ & $0$ & $0$ & $0$ & $0$ & $0$ & $0$ & $0$ & $0$ & $0$ & $0$ & $1$ & $-1$ & $0$ & $0$ & $0$ & $0$ & $0$ \\
$0$ & $0$ & $0$ & $0$ & $0$ & $0$ & $0$ & $0$ & $0$ & $0$ & $0$ & $0$ & $0$ & $1$ & $0$ & $0$ & $0$ & $0$ & $0$ & $0$ & $0$ & $0$ & $0$ & $0$ & $0$ & $-1$ & $0$ & $0$ & $0$ & $0$ & $0$ & $0$ & $0$ & $0$ & $0$ & $0$ & $1$ & $-1$ & $0$ & $0$ & $0$ & $0$ & $0$ & $0$ & $0$ \\
$0$ & $0$ & $0$ & $0$ & $0$ & $0$ & $0$ & $0$ & $0$ & $0$ & $1$ & $0$ & $0$ & $0$ & $0$ & $0$ & $0$ & $0$ & $0$ & $0$ & $0$ & $1$ & $0$ & $0$ & $0$ & $0$ & $0$ & $0$ & $0$ & $0$ & $0$ & $0$ & $0$ & $-1$ & $0$ & $0$ & $0$ & $0$ & $0$ & $0$ & $-1$ & $0$ & $-1$ & $0$ & $0$ \\
$0$ & $0$ & $0$ & $0$ & $0$ & $0$ & $0$ & $1$ & $0$ & $0$ & $0$ & $0$ & $0$ & $0$ & $0$ & $0$ & $-1$ & $0$ & $0$ & $0$ & $0$ & $0$ & $1$ & $0$ & $0$ & $0$ & $-1$ & $0$ & $0$ & $0$ & $0$ & $0$ & $0$ & $0$ & $0$ & $0$ & $0$ & $0$ & $1$ & $0$ & $0$ & $0$ & $0$ & $1$ & $0$ \\
$0$ & $0$ & $0$ & $0$ & $0$ & $0$ & $0$ & $-1$ & $0$ & $0$ & $0$ & $0$ & $0$ & $0$ & $0$ & $0$ & $0$ & $1$ & $0$ & $0$ & $0$ & $0$ & $1$ & $1$ & $0$ & $0$ & $-1$ & $0$ & $0$ & $0$ & $0$ & $0$ & $0$ & $0$ & $0$ & $0$ & $0$ & $0$ & $0$ & $0$ & $-1$ & $0$ & $0$ & $0$ & $0$ \\
$0$ & $0$ & $0$ & $0$ & $0$ & $0$ & $0$ & $0$ & $0$ & $0$ & $0$ & $0$ & $0$ & $0$ & $0$ & $0$ & $0$ & $1$ & $0$ & $0$ & $0$ & $0$ & $0$ & $0$ & $-1$ & $0$ & $0$ & $0$ & $0$ & $0$ & $0$ & $0$ & $0$ & $0$ & $0$ & $0$ & $0$ & $0$ & $0$ & $0$ & $0$ & $0$ & $0$ & $0$ & $-1$ \\
$0$ & $0$ & $0$ & $0$ & $0$ & $0$ & $0$ & $0$ & $0$ & $0$ & $0$ & $0$ & $0$ & $0$ & $0$ & $0$ & $-1$ & $0$ & $0$ & $0$ & $0$ & $0$ & $0$ & $1$ & $-1$ & $0$ & $0$ & $0$ & $0$ & $0$ & $0$ & $0$ & $0$ & $0$ & $0$ & $0$ & $0$ & $0$ & $0$ & $0$ & $0$ & $1$ & $0$ & $0$ & $0$ \\
$0$ & $0$ & $0$ & $0$ & $0$ & $0$ & $0$ & $0$ & $-1$ & $0$ & $0$ & $0$ & $0$ & $0$ & $0$ & $0$ & $0$ & $0$ & $0$ & $0$ & $0$ & $1$ & $0$ & $0$ & $-1$ & $-1$ & $-1$ & $0$ & $0$ & $0$ & $0$ & $0$ & $0$ & $0$ & $1$ & $0$ & $0$ & $0$ & $0$ & $0$ & $0$ & $0$ & $0$ & $0$ & $0$ \\
$0$ & $0$ & $0$ & $0$ & $0$ & $0$ & $0$ & $0$ & $0$ & $0$ & $0$ & $-1$ & $0$ & $0$ & $-1$ & $0$ & $0$ & $0$ & $0$ & $0$ & $0$ & $1$ & $0$ & $0$ & $0$ & $0$ & $-1$ & $0$ & $0$ & $0$ & $0$ & $0$ & $0$ & $0$ & $-1$ & $-1$ & $0$ & $0$ & $0$ & $0$ & $0$ & $0$ & $0$ & $0$ & $0$ \\
$0$ & $1$ & $0$ & $0$ & $0$ & $0$ & $0$ & $0$ & $1$ & $0$ & $0$ & $0$ & $0$ & $0$ & $-1$ & $0$ & $0$ & $0$ & $0$ & $-1$ & $0$ & $-1$ & $0$ & $0$ & $0$ & $1$ & $1$ & $0$ & $0$ & $-1$ & $0$ & $0$ & $0$ & $0$ & $-1$ & $-1$ & $0$ & $0$ & $0$ & $1$ & $0$ & $0$ & $0$ & $0$ & $0$ \\
$0$ & $0$ & $-1$ & $0$ & $0$ & $0$ & $0$ & $0$ & $1$ & $0$ & $0$ & $0$ & $0$ & $0$ & $0$ & $1$ & $0$ & $0$ & $0$ & $0$ & $-1$ & $-1$ & $0$ & $1$ & $0$ & $1$ & $0$ & $0$ & $0$ & $1$ & $0$ & $0$ & $0$ & $0$ & $0$ & $1$ & $0$ & $0$ & $0$ & $0$ & $0$ & $0$ & $0$ & $0$ & $0$ \\
$0$ & $1$ & $0$ & $0$ & $0$ & $0$ & $1$ & $0$ & $0$ & $-1$ & $0$ & $1$ & $-1$ & $0$ & $0$ & $0$ & $0$ & $0$ & $0$ & $-1$ & $0$ & $0$ & $-1$ & $0$ & $-1$ & $0$ & $0$ & $0$ & $0$ & $0$ & $-1$ & $0$ & $0$ & $0$ & $0$ & $0$ & $0$ & $0$ & $0$ & $0$ & $0$ & $0$ & $-1$ & $0$ & $0$ \\
$0$ & $0$ & $0$ & $1$ & $0$ & $0$ & $0$ & $0$ & $0$ & $0$ & $0$ & $0$ & $0$ & $0$ & $0$ & $0$ & $1$ & $-1$ & $0$ & $0$ & $-1$ & $0$ & $0$ & $1$ & $-1$ & $0$ & $-1$ & $0$ & $1$ & $0$ & $0$ & $0$ & $0$ & $1$ & $0$ & $0$ & $0$ & $0$ & $0$ & $0$ & $0$ & $0$ & $0$ & $0$ & $0$ \\
$1$ & $0$ & $0$ & $0$ & $0$ & $0$ & $0$ & $0$ & $0$ & $0$ & $0$ & $-1$ & $0$ & $1$ & $0$ & $0$ & $0$ & $0$ & $0$ & $1$ & $0$ & $0$ & $0$ & $0$ & $0$ & $-1$ & $0$ & $0$ & $0$ & $0$ & $0$ & $-1$ & $0$ & $0$ & $0$ & $-1$ & $0$ & $1$ & $0$ & $0$ & $0$ & $0$ & $0$ & $0$ & $0$ \\
$0$ & $0$ & $0$ & $1$ & $0$ & $0$ & $-1$ & $-1$ & $0$ & $1$ & $0$ & $0$ & $0$ & $0$ & $0$ & $0$ & $1$ & $0$ & $0$ & $-1$ & $0$ & $0$ & $-1$ & $-1$ & $0$ & $0$ & $1$ & $0$ & $0$ & $0$ & $0$ & $0$ & $1$ & $0$ & $0$ & $0$ & $0$ & $0$ & $0$ & $-1$ & $0$ & $0$ & $0$ & $-1$ & $0$ \\
$0$ & $1$ & $0$ & $0$ & $0$ & $0$ & $1$ & $0$ & $0$ & $0$ & $0$ & $0$ & $-1$ & $0$ & $0$ & $0$ & $0$ & $0$ & $1$ & $0$ & $0$ & $0$ & $1$ & $0$ & $0$ & $1$ & $0$ & $-1$ & $0$ & $0$ & $0$ & $0$ & $0$ & $0$ & $0$ & $0$ & $-1$ & $0$ & $1$ & $-1$ & $0$ & $0$ & $0$ & $0$ & $0$ \\
$0$ & $1$ & $0$ & $0$ & $0$ & $0$ & $0$ & $0$ & $0$ & $0$ & $0$ & $1$ & $0$ & $0$ & $0$ & $0$ & $0$ & $0$ & $-1$ & $0$ & $0$ & $-1$ & $0$ & $0$ & $1$ & $0$ & $1$ & $0$ & $-1$ & $0$ & $0$ & $0$ & $0$ & $0$ & $1$ & $0$ & $0$ & $0$ & $0$ & $0$ & $0$ & $0$ & $0$ & $0$ & $0$ \\
$0$ & $0$ & $-1$ & $0$ & $0$ & $0$ & $0$ & $0$ & $0$ & $1$ & $1$ & $0$ & $0$ & $0$ & $0$ & $0$ & $0$ & $0$ & $0$ & $0$ & $-1$ & $0$ & $0$ & $0$ & $0$ & $0$ & $0$ & $0$ & $0$ & $0$ & $1$ & $0$ & $0$ & $1$ & $0$ & $0$ & $-1$ & $0$ & $0$ & $0$ & $0$ & $0$ & $1$ & $0$ & $0$ \\
$1$ & $0$ & $0$ & $0$ & $0$ & $0$ & $0$ & $0$ & $0$ & $0$ & $0$ & $0$ & $0$ & $0$ & $0$ & $0$ & $0$ & $0$ & $0$ & $0$ & $1$ & $0$ & $0$ & $0$ & $0$ & $-1$ & $0$ & $1$ & $0$ & $0$ & $0$ & $0$ & $0$ & $0$ & $0$ & $0$ & $1$ & $-1$ & $0$ & $0$ & $0$ & $0$ & $0$ & $0$ & $-1$ \\
$0$ & $0$ & $-1$ & $0$ & $0$ & $0$ & $1$ & $0$ & $0$ & $0$ & $0$ & $0$ & $0$ & $0$ & $0$ & $0$ & $0$ & $0$ & $-1$ & $0$ & $0$ & $-1$ & $0$ & $1$ & $0$ & $0$ & $0$ & $0$ & $0$ & $0$ & $0$ & $0$ & $-1$ & $1$ & $0$ & $0$ & $0$ & $0$ & $0$ & $0$ & $0$ & $0$ & $0$ & $1$ & $0$ \\
$0$ & $0$ & $0$ & $0$ & $-1$ & $0$ & $0$ & $1$ & $0$ & $0$ & $0$ & $0$ & $-1$ & $0$ & $0$ & $0$ & $0$ & $0$ & $-1$ & $0$ & $0$ & $0$ & $-1$ & $0$ & $0$ & $0$ & $1$ & $0$ & $0$ & $0$ & $0$ & $0$ & $1$ & $0$ & $0$ & $0$ & $0$ & $0$ & $1$ & $-1$ & $1$ & $0$ & $0$ & $0$ & $0$ \\
$0$ & $0$ & $0$ & $0$ & $1$ & $0$ & $0$ & $0$ & $0$ & $0$ & $0$ & $-1$ & $0$ & $0$ & $0$ & $0$ & $0$ & $0$ & $-1$ & $0$ & $0$ & $0$ & $0$ & $0$ & $1$ & $0$ & $0$ & $0$ & $0$ & $0$ & $0$ & $-1$ & $0$ & $0$ & $1$ & $0$ & $0$ & $0$ & $0$ & $0$ & $0$ & $0$ & $0$ & $0$ & $-1$ \\
$0$ & $0$ & $0$ & $0$ & $0$ & $-1$ & $0$ & $0$ & $0$ & $0$ & $0$ & $0$ & $0$ & $1$ & $0$ & $0$ & $0$ & $0$ & $1$ & $0$ & $0$ & $0$ & $0$ & $0$ & $0$ & $0$ & $0$ & $1$ & $0$ & $0$ & $0$ & $0$ & $0$ & $0$ & $0$ & $0$ & $0$ & $0$ & $0$ & $0$ & $0$ & $0$ & $0$ & $1$ & $-1$ \\
$0$ & $0$ & $0$ & $0$ & $0$ & $1$ & $0$ & $1$ & $0$ & $0$ & $-1$ & $0$ & $0$ & $0$ & $0$ & $0$ & $0$ & $0$ & $1$ & $0$ & $0$ & $0$ & $0$ & $-1$ & $0$ & $0$ & $0$ & $0$ & $1$ & $0$ & $0$ & $0$ & $0$ & $1$ & $0$ & $0$ & $0$ & $0$ & $0$ & $0$ & $1$ & $-1$ & $0$ & $0$ & $0$ \\
$0$ & $0$ & $-1$ & $0$ & $0$ & $0$ & $0$ & $0$ & $0$ & $0$ & $1$ & $0$ & $0$ & $-1$ & $0$ & $0$ & $0$ & $0$ & $1$ & $0$ & $0$ & $0$ & $0$ & $0$ & $0$ & $1$ & $0$ & $0$ & $0$ & $0$ & $0$ & $1$ & $0$ & $0$ & $0$ & $0$ & $-1$ & $0$ & $0$ & $0$ & $0$ & $-1$ & $0$ & $0$ & $0$ \\
$0$ & $0$ & $0$ & $0$ & $-1$ & $0$ & $0$ & $0$ & $0$ & $0$ & $0$ & $0$ & $0$ & $0$ & $0$ & $0$ & $0$ & $1$ & $0$ & $0$ & $1$ & $0$ & $1$ & $0$ & $1$ & $0$ & $0$ & $0$ & $0$ & $0$ & $1$ & $0$ & $0$ & $0$ & $0$ & $0$ & $0$ & $-1$ & $0$ & $0$ & $0$ & $0$ & $0$ & $0$ & $-1$ \\
$0$ & $0$ & $0$ & $1$ & $0$ & $0$ & $0$ & $0$ & $-1$ & $0$ & $0$ & $0$ & $0$ & $0$ & $0$ & $0$ & $0$ & $-1$ & $0$ & $1$ & $0$ & $0$ & $0$ & $0$ & $1$ & $0$ & $0$ & $0$ & $0$ & $0$ & $0$ & $-1$ & $0$ & $0$ & $1$ & $0$ & $0$ & $0$ & $0$ & $0$ & $0$ & $1$ & $0$ & $0$ & $0$ \\
$0$ & $0$ & $0$ & $1$ & $0$ & $0$ & $0$ & $0$ & $-1$ & $1$ & $0$ & $0$ & $0$ & $0$ & $0$ & $0$ & $0$ & $0$ & $0$ & $0$ & $-1$ & $0$ & $-1$ & $0$ & $0$ & $0$ & $0$ & $-1$ & $0$ & $0$ & $0$ & $0$ & $0$ & $0$ & $0$ & $0$ & $0$ & $0$ & $1$ & $0$ & $0$ & $0$ & $0$ & $0$ & $0$ \\
$0$ & $0$ & $0$ & $0$ & $-1$ & $0$ & $0$ & $0$ & $0$ & $0$ & $0$ & $0$ & $0$ & $0$ & $-1$ & $0$ & $1$ & $0$ & $0$ & $0$ & $-1$ & $0$ & $0$ & $0$ & $0$ & $0$ & $-1$ & $0$ & $0$ & $-1$ & $0$ & $0$ & $0$ & $0$ & $0$ & $0$ & $0$ & $0$ & $1$ & $0$ & $0$ & $0$ & $0$ & $0$ & $0$ \\
$0$ & $0$ & $0$ & $0$ & $0$ & $-1$ & $0$ & $0$ & $0$ & $0$ & $0$ & $0$ & $0$ & $0$ & $0$ & $0$ & $0$ & $-1$ & $0$ & $-1$ & $0$ & $0$ & $0$ & $0$ & $0$ & $0$ & $0$ & $0$ & $0$ & $0$ & $1$ & $0$ & $0$ & $0$ & $0$ & $0$ & $0$ & $-1$ & $0$ & $0$ & $-1$ & $0$ & $0$ & $0$ & $0$ \\
$0$ & $0$ & $0$ & $0$ & $0$ & $1$ & $0$ & $1$ & $0$ & $0$ & $0$ & $0$ & $0$ & $-1$ & $0$ & $-1$ & $-1$ & $0$ & $0$ & $-1$ & $0$ & $0$ & $0$ & $-1$ & $0$ & $0$ & $0$ & $0$ & $0$ & $-1$ & $0$ & $0$ & $0$ & $0$ & $0$ & $0$ & $0$ & $0$ & $0$ & $0$ & $0$ & $1$ & $0$ & $-1$ & $0$ \\
$1$ & $0$ & $0$ & $0$ & $0$ & $0$ & $0$ & $0$ & $0$ & $0$ & $1$ & $0$ & $0$ & $0$ & $1$ & $-1$ & $0$ & $0$ & $0$ & $0$ & $1$ & $1$ & $0$ & $0$ & $0$ & $0$ & $0$ & $0$ & $-1$ & $0$ & $0$ & $0$ & $0$ & $0$ & $0$ & $-1$ & $0$ & $0$ & $0$ & $0$ & $0$ & $0$ & $1$ & $0$ & $0$ \\
$1$ & $0$ & $0$ & $0$ & $0$ & $0$ & $0$ & $0$ & $0$ & $0$ & $0$ & $0$ & $-1$ & $0$ & $1$ & $-1$ & $0$ & $0$ & $0$ & $1$ & $0$ & $1$ & $0$ & $0$ & $0$ & $0$ & $0$ & $0$ & $0$ & $0$ & $0$ & $0$ & $-1$ & $0$ & $0$ & $0$ & $0$ & $0$ & $0$ & $0$ & $-1$ & $0$ & $-1$ & $0$ & $0$ \\
\specialrule{1pt}{0pt}{0pt}
\multicolumn{18}{c}{$\phi$} & \multicolumn{9}{c}{$\psi$} & \multicolumn{18}{c}{$\upsilon^{-T}$}\\
\specialrule{1pt}{0pt}{0pt}
$-\frac{1}{8}$ & $\frac{1}{8}$ & $0$ & $-1$ & $0$ & $-1$ & $-1$ & $0$ & $1$ & $1$ & $0$ & $0$ & $\frac{1}{8}$ & $-\frac{1}{8}$ & $0$ & $-1$ & $0$ & $0$ & $0$ & $0$ & $0$ & $0$ & $1$ & $-1$ & $0$ & $-1$ & $1$ & $0$ & $0$ & $1$ & $\frac{1}{8}$ & $0$ & $0$ & $\frac{1}{8}$ & $0$ & $0$ & $0$ & $-\frac{1}{8}$ & $0$ & $0$ & $-2$ & $1$ & $0$ & $\frac{1}{8}$ & $0$ \\
$0$ & $0$ & $0.25$ & $0$ & $0$ & $2$ & $-1$ & $1$ & $-1$ & $-1$ & $1$ & $1$ & $-\frac{1}{8}$ & $\frac{1}{8}$ & $-\frac{1}{8}$ & $0$ & $0$ & $0$ & $0$ & $0$ & $0$ & $1$ & $-1$ & $0$ & $-1$ & $1$ & $0$ & $1$ & $1$ & $1$ & $\frac{1}{8}$ & $\frac{1}{8}$ & $\frac{1}{8}$ & $0$ & $0$ & $0$ & $0$ & $0$ & $0$ & $0$ & $0$ & $0$ & $-\frac{1}{8}$ & $-\frac{1}{8}$ & $-\frac{1}{8}$ \\
$0$ & $0$ & $-\frac{1}{8}$ & $0$ & $-1$ & $0$ & $0$ & $-1$ & $0$ & $1$ & $0$ & $0$ & $0$ & $0$ & $-\frac{1}{8}$ & $-1$ & $0$ & $2$ & $1$ & $-1$ & $0$ & $0$ & $0$ & $0$ & $-1$ & $1$ & $0$ & $0$ & $0$ & $1$ & $\frac{1}{8}$ & $0.25$ & $0$ & $\frac{1}{8}$ & $0$ & $0$ & $0$ & $\frac{1}{8}$ & $0$ & $0$ & $0$ & $1$ & $0$ & $\frac{1}{8}$ & $0$ \\
$0$ & $0$ & $-\frac{1}{8}$ & $2$ & $1$ & $0$ & $0$ & $-1$ & $0$ & $1$ & $0$ & $0$ & $0$ & $0$ & $\frac{1}{8}$ & $-1$ & $0$ & $0$ & $-1$ & $0$ & $-1$ & $0$ & $1$ & $1$ & $1$ & $-1$ & $0$ & $0$ & $0$ & $0$ & $\frac{1}{8}$ & $-\frac{1}{8}$ & $-\frac{1}{8}$ & $0$ & $0$ & $0$ & $-\frac{1}{8}$ & $\frac{1}{8}$ & $\frac{1}{8}$ & $-1$ & $1$ & $1$ & $0$ & $0$ & $0$ \\
$0$ & $0$ & $0$ & $0$ & $0$ & $0$ & $1$ & $-1$ & $1$ & $1$ & $-1$ & $1$ & $\frac{1}{8}$ & $-\frac{1}{8}$ & $\frac{1}{8}$ & $0$ & $0$ & $0$ & $-1$ & $0$ & $1$ & $0$ & $-1$ & $1$ & $1$ & $-1$ & $0$ & $1$ & $1$ & $0$ & $0$ & $\frac{1}{8}$ & $\frac{1}{8}$ & $0$ & $\frac{1}{8}$ & $\frac{1}{8}$ & $0$ & $\frac{1}{8}$ & $0$ & $-1$ & $1$ & $2$ & $0$ & $-\frac{1}{8}$ & $0$ \\
$0$ & $0$ & $0$ & $-1$ & $1$ & $-1$ & $0$ & $0$ & $0$ & $0$ & $0$ & $0$ & $-\frac{1}{8}$ & $\frac{1}{8}$ & $-\frac{1}{8}$ & $1$ & $-1$ & $1$ & $0$ & $0$ & $0$ & $0$ & $2$ & $0$ & $0$ & $-2$ & $0$ & $0$ & $0$ & $-1$ & $-\frac{1}{8}$ & $0$ & $0$ & $-\frac{1}{8}$ & $0$ & $0$ & $0$ & $\frac{1}{8}$ & $0$ & $2$ & $0$ & $1$ & $0$ & $-\frac{1}{8}$ & $0$ \\
$0$ & $0$ & $0$ & $1$ & $-1$ & $-1$ & $1$ & $-1$ & $-1$ & $1$ & $-1$ & $-1$ & $0$ & $0$ & $0$ & $-1$ & $1$ & $1$ & $1$ & $0$ & $-1$ & $0$ & $1$ & $-1$ & $1$ & $-1$ & $0$ & $0$ & $0$ & $-1$ & $-\frac{1}{8}$ & $0$ & $0$ & $0$ & $\frac{1}{8}$ & $\frac{1}{8}$ & $0$ & $\frac{1}{8}$ & $0$ & $1$ & $-1$ & $0$ & $\frac{1}{8}$ & $0$ & $\frac{1}{8}$ \\
$0$ & $0$ & $0$ & $1$ & $-1$ & $1$ & $1$ & $-1$ & $1$ & $1$ & $-1$ & $1$ & $0$ & $0$ & $0$ & $-1$ & $1$ & $-1$ & $1$ & $0$ & $-1$ & $0$ & $-1$ & $1$ & $-1$ & $1$ & $0$ & $1$ & $1$ & $0$ & $\frac{1}{8}$ & $0$ & $0$ & $0$ & $\frac{1}{8}$ & $\frac{1}{8}$ & $\frac{1}{8}$ & $0$ & $-\frac{1}{8}$ & $0$ & $0$ & $1$ & $0$ & $-\frac{1}{8}$ & $0$ \\
$0$ & $0$ & $-\frac{1}{8}$ & $1$ & $0$ & $-1$ & $0$ & $-1$ & $0$ & $1$ & $0$ & $0$ & $-\frac{1}{8}$ & $-\frac{1}{8}$ & $0$ & $0$ & $1$ & $1$ & $-2$ & $0$ & $0$ & $0$ & $0$ & $0$ & $0$ & $-2$ & $0$ & $-1$ & $-1$ & $-1$ & $-\frac{1}{8}$ & $-\frac{1}{8}$ & $-\frac{1}{8}$ & $-\frac{1}{8}$ & $-\frac{1}{8}$ & $-\frac{1}{8}$ & $0$ & $0$ & $0$ & $-1$ & $-1$ & $-1$ & $0$ & $0$ & $0$ \\
$-\frac{1}{8}$ & $-\frac{1}{8}$ & $0$ & $0$ & $-1$ & $0$ & $1$ & $0$ & $-1$ & $1$ & $0$ & $0$ & $0$ & $0$ & $-\frac{1}{8}$ & $0$ & $1$ & $1$ &   &   &   &   &   &   &   &   &   & $0$ & $0$ & $-1$ & $0$ & $\frac{1}{8}$ & $-\frac{1}{8}$ & $-\frac{1}{8}$ & $0$ & $0$ & $-\frac{1}{8}$ & $0$ & $\frac{1}{8}$ & $-1$ & $1$ & $0$ & $0$ & $-\frac{1}{8}$ & $0$ \\
$0$ & $0$ & $-\frac{1}{8}$ & $-1$ & $0$ & $1$ & $0$ & $-1$ & $0$ & $1$ & $0$ & $0$ & $-\frac{1}{8}$ & $\frac{1}{8}$ & $0$ & $0$ & $-1$ & $1$ &   &   &   &   &   &   &   &   &   & $1$ & $-1$ & $-1$ & $0$ & $0$ & $0$ & $0$ & $0$ & $0$ & $-\frac{1}{8}$ & $\frac{1}{8}$ & $\frac{1}{8}$ & $-1$ & $1$ & $1$ & $\frac{1}{8}$ & $-\frac{1}{8}$ & $-\frac{1}{8}$ \\
$0$ & $0$ & $-\frac{1}{8}$ & $-1$ & $0$ & $-1$ & $0$ & $-1$ & $0$ & $1$ & $0$ & $0$ & $\frac{1}{8}$ & $-\frac{1}{8}$ & $0$ & $0$ & $-1$ & $-1$ &   &   &   &   &   &   &   &   &   & $-1$ & $-1$ & $1$ & $\frac{1}{8}$ & $\frac{1}{8}$ & $-\frac{1}{8}$ & $\frac{1}{8}$ & $\frac{1}{8}$ & $-\frac{1}{8}$ & $0$ & $0$ & $0$ & $-1$ & $-1$ & $1$ & $0$ & $0$ & $0$ \\
$0$ & $0$ & $\frac{1}{8}$ & $0$ & $-1$ & $0$ & $-1$ & $0$ & $-1$ & $0$ & $1$ & $1$ & $-\frac{1}{8}$ & $\frac{1}{8}$ & $0$ & $-1$ & $0$ & $0$ &   &   &   &   &   &   &   &   &   & $-1$ & $-1$ & $0$ & $0$ & $\frac{1}{8}$ & $-\frac{1}{8}$ & $\frac{1}{8}$ & $0$ & $0$ & $0$ & $\frac{1}{8}$ & $0$ & $0$ & $0$ & $1$ & $-\frac{1}{8}$ & $0$ & $-\frac{1}{8}$ \\
$0$ & $0$ & $-\frac{1}{8}$ & $0$ & $1$ & $0$ & $1$ & $0$ & $-1$ & $0$ & $-1$ & $1$ & $-\frac{1}{8}$ & $\frac{1}{8}$ & $0$ & $1$ & $0$ & $0$ &   &   &   &   &   &   &   &   &   & $0$ & $0$ & $1$ & $0$ & $-\frac{1}{8}$ & $-\frac{1}{8}$ & $\frac{1}{8}$ & $0$ & $0$ & $\frac{1}{8}$ & $0$ & $\frac{1}{8}$ & $-1$ & $1$ & $0$ & $0$ & $\frac{1}{8}$ & $0$ \\
$\frac{1}{8}$ & $-\frac{1}{8}$ & $\frac{1}{8}$ & $1$ & $-1$ & $1$ & $-1$ & $1$ & $-1$ & $0$ & $0$ & $0$ & $-\frac{1}{8}$ & $\frac{1}{8}$ & $-\frac{1}{8}$ & $0$ & $0$ & $0$ &   &   &   &   &   &   &   &   &   & $1$ & $1$ & $0$ & $0$ & $\frac{1}{8}$ & $\frac{1}{8}$ & $\frac{1}{8}$ & $0$ & $0$ & $0$ & $\frac{1}{8}$ & $0$ & $0$ & $0$ & $1$ & $-\frac{1}{8}$ & $0$ & $\frac{1}{8}$ \\
$-\frac{1}{8}$ & $\frac{1}{8}$ & $0$ & $0$ & $-1$ & $0$ & $-1$ & $0$ & $1$ & $1$ & $0$ & $0$ & $0$ & $0$ & $-\frac{1}{8}$ & $0$ & $-1$ & $1$ &   &   &   &   &   &   &   &   &   & $1$ & $-1$ & $0$ & $0$ & $\frac{1}{8}$ & $\frac{1}{8}$ & $-\frac{1}{8}$ & $0$ & $0$ & $0$ & $-\frac{1}{8}$ & $0$ & $0$ & $0$ & $-1$ & $-\frac{1}{8}$ & $0$ & $-\frac{1}{8}$ \\
$0$ & $0$ & $\frac{1}{8}$ & $0$ & $-1$ & $0$ & $1$ & $0$ & $1$ & $0$ & $-1$ & $1$ & $-\frac{1}{8}$ & $-\frac{1}{8}$ & $0$ & $-1$ & $0$ & $0$ &   &   &   &   &   &   &   &   &   & $0$ & $0$ & $-1$ & $0$ & $-\frac{1}{8}$ & $\frac{1}{8}$ & $-\frac{1}{8}$ & $0$ & $0$ & $\frac{1}{8}$ & $0$ & $\frac{1}{8}$ & $1$ & $1$ & $0$ & $0$ & $-\frac{1}{8}$ & $0$ \\
$-\frac{1}{8}$ & $-\frac{1}{8}$ & $0$ & $1$ & $0$ & $1$ & $0$ & $-1$ & $0$ & $0$ & $-1$ & $1$ & $0$ & $0$ & $\frac{1}{8}$ & $-1$ & $0$ & $0$ &   &   &   &   &   &   &   &   &   & $0$ & $0$ & $1$ & $\frac{1}{8}$ & $0$ & $0$ & $0$ & $-\frac{1}{8}$ & $\frac{1}{8}$ & $0$ & $-\frac{1}{8}$ & $0$ & $1$ & $-1$ & $0$ & $-\frac{1}{8}$ & $0$ & $\frac{1}{8}$ \\
\specialrule{1pt}{0pt}{0pt}
\end{tabular}
}
\end{center}

\end{document}